\documentclass[11pt,draft]{article}

\usepackage[T1]{fontenc}

\usepackage{amsmath,amssymb,amsthm}
\usepackage[active]{srcltx}

\newcommand{\R}{\mathbb{R}}
\newcommand{\C}{\mathbb{C}}

\newcommand{\A}{{\cal A}}
\newcommand{\D}{{\cal D}}
\newcommand{\Abar}{\overline{\cal A}}
\newcommand{\K}{{\cal K}}
\newcommand{\F}{{\cal F}}
\newcommand{\pr}{{\rm pr}}
\newcommand{\la}{\lambda}

\newcommand{\h}{{\cal H}}

\newcommand{\eps}{\epsilon}
\newcommand{\veps}{\varepsilon}
\newcommand{\Cyl}{{\rm Cyl}}
\newcommand{\ot}{\otimes}

\newcommand{\bld}[1]{\boldsymbol{#1}}

\newcommand{\diff}{{\rm Diff}}

\newcommand{\we}{\wedge}

\newcommand{\lr}{\lrcorner\,}

\DeclareMathOperator{\tr}{{\rm tr}}
\DeclareMathOperator{\sgn}{{\rm sgn}}
\DeclareMathOperator{\spn}{{\rm span}}
\newcounter{mnotecount}[section]

\newtheorem{thr}{Theorem}
\newtheorem{lm}[thr]{Lemma}
\newtheorem{df}[thr]{Definition}
\newtheorem{cor}[thr]{Corollary}
\newtheorem{pro}[thr]{Proposition}

\numberwithin{equation}{section}
\numberwithin{thr}{section}

\begin{document}

\title{Kinematic quantum states for the Teleparallel Equivalent of General Relativity\footnote{This is an author-created version of a paper published as {\em Gen. Rel. Grav.} {\bf 46} 1653 (2014) DOI 10.1007/s10714-013-1653-3}}
\author{Andrzej Oko{\l}\'ow}
\date{July 14, 2014}

\maketitle
\begin{center}
{\it  Institute of Theoretical Physics, Warsaw University\\ ul. Pasteura 5, 02-093 Warsaw, Poland\smallskip\\
oko@fuw.edu.pl}
\end{center}
\medskip

\begin{abstract}
A space of kinematic quantum states for the Teleparallel Equivalent of General Relativity is constructed by means of projective techniques. The states are kinematic in this sense that their construction bases merely on the structure of the phase space of the theory and does not take into account constraints on it. The space of quantum states is meant to serve as an element of a canonical background independent quantization of the theory.
\end{abstract}

\section{Introduction}

Nowadays there are many approaches \cite{app,carlip} to quantum gravity but so far no one is fully successful. Therefore it is still worth to take a risk to develop a new approach. It seems that the Teleparallel Equivalent of General Relativity (TEGR) was never used as a point of departure for a construction of a model of quantum gravity and therefore we would like to check whether it is possible to quantize gravity in this formulation (for the latest review of TEGR see \cite{mal-rev}). More precisely, we would like to check whether it is possible to quantize TEGR using the method of canonical quantization or, if it is needed, a modification of the method. Since TEGR is a background independent theory we would like to quantize it in a background independent manner.

TEGR in its canonical formulation is a constrained system (see e.g. \cite{maluf-1,maluf,oko-tegr,bl}). Therefore  it is quite natural to attempt to apply the Dirac strategy of canonical quantization of such systems which requires two steps to be carried out: $(i)$ first one neglects constraints and constructs a space of kinematic quantum states, that is, quantum states corresponding to all classical states constituting the whole phase space $(ii)$ then among the kinematic quantum states  one distinguishes physical quantum states as those corresponding to classical states satisfying all the constraints. The space of kinematic quantum states is usually a Hilbert space and to carry out the second step one tries to find operators on the Hilbert space corresponding to the constraints and singles out physical quantum states as those annihilated by the operators (this procedure is valid if all the constraints are of the first class).    

In this paper we construct a space of kinematic quantum states for TEGR treated as a theory of cotetrad fields on a four-dimensional manifold. More precisely, the construction is valid for any theory of cotetrad fields the phase space of which coincides with that of TEGR---an example of such a theory is the Yang-Mills-type Teleparallel Model (YMTM) considered in \cite{itin,os}. 

The space of quantum states for TEGR, which since now will be denoted by $\D$, will be constructed according to a method presented in \cite{q-stat} combined with some Loop Quantum Gravity (LQG) techniques \cite{acz,cq-diff,rev,rev-1}. This method being a generalization of a construction by Kijowski \cite{kpt} provides us with a space of quantum states which is not a Hilbert space but rather {\em a convex set of quantum states}---these states can be seen as algebraic states (i.e. linear positive normed functionals) on a $C^*$-algebra which can be thought of as an algebra of some quantum observables.    

We will also show that spatial diffeomorphisms act naturally on the space $\D$ which allows to hope that $\D$ can be used as an element of a background independent quantization of TEGR.  

The construction of $\D$ is similar to a construction of a space of quantum states for the degenerate Pleba\'nski gravity (DPG) \cite{q-stat} and the descriptions of both constructions follow the same pattern. It may be helpful to study first the construction in \cite{q-stat} since it is simpler than that of $\D$.

Let us mention that except the space $\D$ it is possible to construct other spaces of kinematical quantum states for TEGR---in this paper we will briefly describe the other spaces and comment on their possible application to quantization of TEGR.

To proceed further with quantization of TEGR it is necessary to single out physical quantum states in the space $\D$, that is, to carry out the second step of the Dirac strategy. Since $\D$ is not a Hilbert space the standard procedure mentioned above by means of which one distinguishes physical quantum states has to be modified in a way. At this moment we are not able to present a satisfactory and workable modification of the procedure (some remarks on this very important issue can be found in \cite{q-stat}), but we hope that this problem will be solved in the future.

The paper is organized as follows: Section 2 contains preliminaries, in Section 3 the space of quantum states for TEGR is constructed, in Section 4 we define an action of spatial diffeomorphisms on $\D$, Section 5 contains a short description of the other spaces of quantum states, in Section 6 we discuss the results. Finally, in Appendix A we show that the space $\D$ is identical to one of the other spaces.

\section{Preliminaries}

\subsection{Cotetrad fields}

Let $\mathbb{M}$ be a real four-dimensional oriented vector space equipped with a scalar product $\eta$ of signature $(-,+,+,+)$. We fix an orthonormal basis $(v_A)$ $(A=0,1,2,3)$ of $\mathbb{M}$ such that the components $(\eta_{AB})$ of $\eta$ given by the basis form the matrix ${\rm diag}(-1,1,1,1)$. The matrix $(\eta_{AB})$ and its inverse $(\eta^{AB})$ will be used to, respectively, lower and raise capital Latin letter indeces $A,B,C,D\in\{0,1,2,3\}$. 

Denote by $\mathbb{E}$ the subspace of $\mathbb{M}$ spanned by the vectors $\{v_1,v_2,v_3\}$. The scalar product $\eta$ induces on $\mathbb{E}$ a positive definite scalar product $\delta$. Its components $(\delta_{IJ})$ in the basis $(v_1,v_2,v_3)$ form a matrix ${\rm diag}(1,1,1)$. The matrix $(\delta_{IJ})$ and its inverse $(\delta^{IJ})$ will be used to, respectively, lower and raise capital Latin letter indeces $I,J,K,L,M\in\{1,2,3\}$. In some formulae we will use the three-dimensional permutation symbol which will be denoted by $\veps_{IJK}$.

\subsection{Phase space \label{phsp}}

The goal of this paper is to construct a space of quantum states for theories of a particular phase space consisting of some fields defined on a three-dimensional {\em oriented} manifold $\Sigma$---a point in the phase space consists of:
\begin{enumerate}
\item a quadruplet of one-forms $(\theta^{A})$, $A=0,1,2,3$,  on $\Sigma$ such that the metric 
\begin{equation}
q:=\eta_{AB}\theta^A\ot\theta^B
\label{q}
\end{equation}
is Riemannian (positive definite); 
\item a quadruplet of two-forms $(p_B)$, $B=0,1,2,3$,  on $\Sigma$.  
\end{enumerate}         
$p_A$ is the momentum conjugate to $\theta^A$. The set of all $(\theta^A)$ satisfying the assumption above will be called a {\em Hamiltonian configuration space} and denoted by $\Theta$, while the set of all $(p_A)$ will be called a {\em momentum space} and denoted by $P$. Thus the phase space is the Cartesian product $P\times\Theta$. The Poisson bracket between two functions $f_1$ and $f_2$ on the phase space is given by the following formula
\begin{equation}
\{f_1,f_2\}=\int_\Sigma\Big(\frac{\delta f_1}{\delta {\theta}^A}\we\frac{\delta f_2}{\delta p_A}-\frac{\delta f_2}{\delta {\theta}^A}\we\frac{\delta f_1}{\delta p_A}\Big)
\label{poiss-0}
\end{equation}
---a definition of the variational derivative with respect to a differential form can be found in \cite{os}.      

As shown in, respectively, \cite{oko-tegr} and \cite{os} both TEGR and YMTM  possess such a phase space.  

It turns out \cite{q-suit} that it is possible to construct quantum states via the method presented in \cite{q-stat} starting from the phase space description above (which in a sense is a natural description)---see Section \ref{other}. However, as it was argued in \cite{q-suit}, a space of these quantum states possesses an undesired property. Therefore the space of quantum states $\D$ will be constructed starting from another description \cite{ham-nv} of the phase space. 

Let $\iota$ be a function defined on a space of all global coframes on $\Sigma$ valued in $\{-1,1\}$. Since for every $(\theta^A)=(\theta^0,\theta^I)\in\Theta$ the triplet $(\theta^I)$ is a global coframe on the manifold \cite{q-suit} $\iota$ can be regarded as a function on $\Theta$. Every function $\iota$ which is a constant function on every path-connected subset of $\Theta$  defines new variables on the phase space \cite{ham-nv} which provide new description of the space. According to it a point in the phase space consists of:
\begin{enumerate}
\item a collection $(\xi^I_\iota,\theta^J)\equiv \theta$, where $\xi_\iota^I$, $I=1,2,3$, is a real function (a zero-form) on $\Sigma$ and $(\theta^J)$, $J=1,2,3$, are one-forms on $\Sigma$ constituting a global coframe;
\item a collection $(\zeta_{\iota I},r_J)\equiv p$, where $\zeta_{\iota I}$, $I=1,2,3$, is a three-form on $\Sigma$ and $r_J$, $J=1,2,3$, is a two-form on the manifold. 
\end{enumerate}                  
$\zeta_{\iota I}$ is the momentum conjugate to $\xi^I_\iota$ and $r_J$ is the momentum conjugate to $\theta^J$. Thus all the $(\xi^I_\iota,\theta^J)$ constitute the Hamiltonian configuration space $\Theta$ while all the $(\zeta_{\iota I},r_J)$ constitute the momentum space $P$. The Poisson \eqref{poiss-0} reads now as follows 
\begin{equation}
\{f_1,f_2\}=\int_\Sigma\Big(\frac{\delta f_1}{\delta {\xi}_\iota^I}\we\frac{\delta f_2}{\delta \zeta_{\iota I}}+\frac{\delta f_1}{\delta {\theta}^I}\we\frac{\delta f_2}{\delta r_I}-\frac{\delta f_2}{\delta {\xi}_\iota^I}\we\frac{\delta f_1}{\delta \zeta_{\iota I}}-\frac{\delta f_2}{\delta {\theta}^I}\we\frac{\delta f_1}{\delta r_I}\Big).
\label{poiss}
\end{equation}

Regarding a relation of the latter description to the former let us first express the dependence of $(p_A,\theta^B)$ on $(\zeta_{\iota I},r_J,\xi_\iota^K,\theta^L)$ \cite{ham-nv}:
\begin{equation}
\begin{aligned}
&p_0=\iota(\theta^K)\sqrt{1+\xi_{\iota J}\xi_\iota^J}\,\vec{\theta}^I\lr\zeta_{\iota I},&& p_I=r_I-\xi_{\iota I}\,\vec{\theta}^J\lr\zeta_{\iota J},\\
&\theta^0=\iota(\theta^J)\frac{\xi_{\iota I}}{\sqrt{1+\xi_{\iota K}\xi_\iota^K}}\,\theta^I, & & \theta^I=\theta^I.
\end{aligned} 
\label{old-new}
\end{equation}
Here $\vec{\theta}^I$ is a vector field on $\Sigma$ obtained from $\theta^I$ by raising its index by a metric inverse to the metric $q$---in a local coordinate frame $(x^i)$ on $\Sigma$       
\[
\vec{\theta}^I:=q^{ij}\theta^I_j\partial_{x^i}.
\]
Since \cite{ham-nv}
\begin{equation}
q=\Big(\delta_{IJ}-\frac{\xi_{\iota I}\xi_{\iota J}}{1+\xi_{\iota K}\xi_\iota^K}\Big)\theta^I\ot\theta^J
\label{q-xi}
\end{equation}
the vector field $\vec{\theta}^I$ is a function of both $\xi_\iota^J$ and $\theta^L$.

The inverse dependence, that is, the dependence of $(\zeta_{\iota I},r_J,\xi_\iota^K,\theta^L)$ on $(p_A,\theta^B)$ reads \cite{ham-nv}
\begin{equation}
\begin{aligned}
&\zeta_{\iota I}=\iota(\theta^K)\sqrt{\det (q_{MN})}\,q_{IJ}\,\theta^J\we p_0, \\
&r_I= \frac{\sqrt{\det (q_{MN})}}{2}\sgn(\theta^L)*(\theta^0\we\theta^J\we\theta^K)\,\veps_{IJK}\,p_0+p_I,\\
&\xi^I_\iota=\frac{1}{2}\frac{\iota(\theta^L)}{\sgn(\theta^L)}*(\theta^0\we\theta_J\we\theta_K)\,\veps^{IJK},\\
&\theta^I=\theta^I.
\end{aligned}
\label{new-old}
\end{equation}
Here $*$ is a Hodge operator defined by the metric $q$, and $(q_{IJ})$ are components of $q$ in the basis $(\theta^J)$. Let us emphasize that in \eqref{new-old} $q$ is treated as a function of $(\theta^A)$ (see \eqref{q}). Moreover,  
\begin{equation}
\sgn(\theta^I):=
\begin{cases}
1 & \text{if $(\theta^I)$ is compatible with the orientation of $\Sigma$}\\
-1 & \text{otherwise}
\end{cases}.
\label{sgn-th}
\end{equation}

\section{Construction of quantum states for a theory of the phase space $P\times\Theta$}

\subsection{Choice of variables}

The construction of a space of quantum states for TEGR we are going to present in this section can be successfully carried out starting from any variables $(\zeta_{\iota I},r_J,\xi_\iota^K,\theta^L)$. However, as proved in \cite{ham-nv} unless $\iota=\sgn$ or $\iota=-\sgn$, where $\sgn$ is given by \eqref{sgn-th}, the constraints of TEGR found in \cite{oko-tegr} and the constraints of YMTM found in \cite{os} cannot be imposed on the resulting space of quantum states. Therefore it is reasonable to restrict ourselves to variables
\begin{align*}
&(\zeta_{s I},r_J,\xi_s^K,\theta^L), &&(\zeta_{-s I},r_J,\xi_{-s}^K,\theta^L)
\end{align*}   
defined by, respectively, $\iota=\sgn$ or $\iota=-\sgn$. Actually, we will construct the space $\D$ using the variables $(\zeta_{s I},r_J,\xi_s^K,\theta^L)$, and then we  will show that a space $\D_{-s}$ built from the variables $(\zeta_{-s I},r_J,\xi_{-s}^K,\theta^L)$  coincides with $\D$.

Since now we will use a simplified notation according to which 
\begin{equation}
(\zeta_{I},r_J,\xi^K,\theta^L)\equiv(\zeta_{s I},r_J,\xi_{s}^K,\theta^L).
\label{simp-n}
\end{equation}

\subsection{Outline of the construction \label{out}}

Following \cite{q-stat} we will first choose $(i)$ a special set $\cal K$ of real functions on $\Theta$ and call the functions {\em configurational elementary degrees of freedom} and $(ii)$ a special set $\cal F$ of real functions on $P$ and call the functions  {\em momentum elementary degrees of freedom}. The configurational d.o.f. will be then used to define functions on $\Theta$ of a special sort called {\em cylindrical functions}. Next, each momentum d.o.f. will define via the Poisson bracket \eqref{poiss} or its regularization a linear operator on the space of cylindrical functions. Thus we will obtain a {\em linear space} $\hat{\cal F}$ spanned by operators associated with elements of $\cal F$.

In the next step of the construction we will choose a set $\Lambda$ such that each element of it is a pair $(\hat{F},K)$, where $\hat{F}$ is a finite dimensional linear subspace of $\hat{\cal F}$ and $K$ is a finite set of configurational elementary d.o.f.. Then we will define on $\Lambda$ a relation $\geq$ equipping it with the structure of a directed set and show that $(\Lambda,\geq)$ satisfies some special assumptions. This will finish the construction since at this moment we will refer to \cite{q-stat} where it was shown that from each directed set satisfying the assumptions one can build a space of quantum states. 

The construction of the space of quantum states from such a directed set $(\Lambda,\geq)$ proceeds as follows. Given $(\hat{F},K)\equiv\lambda\in \Lambda$, one uses elements of $K$ to reduce the ``infinite-dimensional'' space $\Theta$ to a space $\Theta_K$ of finite dimension. Next, one defines $(i)$ a Hilbert space $\h_\lambda$  as a space of functions on $\Theta_K$ square integrable with respect to a natural measure on $\Theta_K$ and $(ii)$ a space $\D_\lambda$ of all density operators on the Hilbert space (i.e. positive operators of trace equal $1$). It turns out that assumed properties of the set $(\Lambda,\geq)$ unambiguously induce on a set $\{\D_\lambda\}_{\lambda\in\Lambda}$ the structure of a projective family. The space of quantum states is then defined as the projective limit of the family.        

Let us emphasize that our choice of elementary d.o.f. as well as application of graphs, cylindrical functions  and the operators defined on them by the Poisson bracket is motivated by LQG methods---see \cite{acz,cq-diff,rev,rev-1} and references therein. 

\subsection{Submanifolds of $\Sigma$}

Each elementary d.o.f. we are going to use will be associated with a submanifold of $\Sigma$. 

Following the LQG methods since now till the end of this paper we will assume that the manifold $\Sigma$ is {\em real analytic}\footnote{Equally well we could assume that the manifold is {\em semi-analytic}---see e.g. \cite{lost,fl} for the definition of semi-analyticity.}.  

An {\em analytic edge} is a one-dimensional connected analytic embedded submanifold of $\Sigma$ with two-point boundary. An {\em oriented} one-dimensional connected $C^0$ submanifold of $\Sigma$ given by a finite union of analytic edges will be called an {\em edge}. The set of all edges in $\Sigma$ will be denoted by $\cal E$.  

Given an edge $e$ of two-point boundary, its orientation  allows to call one of its endpoints {\em a source} and the other {\em a target} of the edge; if an edge is a loop then we distinguish one of its points and treat it simultaneously as the source and the target of the edge.

An edge $e^{-1}$ is called an {\em inverse} of an edge $e$ if $e^{-1}$ and $e$ coincide as un-oriented submanifolds of $\Sigma$ and differ by their orientations.  We say that an edge $e$ is a composition of the edges $e_1$ and $e_2$, $e=e_2\circ e_1$, if $(i)$ $e$ as an oriented manifold is a union of $e_1$ and $e_2$, $(ii)$ the target of $e_1$ coincides with the source of $e_2$ and $(iii)$ $e_1\cap e_2$ consists solely of some (or all) endpoints of $e_1$ and $e_2$.               

We say that two edges are {\em independent} if the set of their common points is either empty or consists solely of some (or all) endpoints of the edges. A {\em graph} in $\Sigma$ is a finite set of pairwise independent edges. Any finite set of edges can be described in terms of edges of a graph \cite{al-hoop}: 
\begin{lm}
For every finite set $E=\{e_1,\ldots,e_N\}$ of edges there exists a graph $\gamma$ in $\Sigma$ such that every $e_j\in E$ is a composition of some edges of $\gamma$ and the inverses of some edges of the graph.  
\label{E-gamma}
\end{lm}

The set of all graphs in $\Sigma$ is naturally a directed set: $\gamma'\geq\gamma$ if each edge of the graph $\gamma$ is a composition of some edges of the graph $\gamma'$ and the inverses of some edges of $\gamma'$. 
 
Let $S$ be a two-dimensional embedded submanifold of $\Sigma$. Assume that $S$ is $(i)$ analytic, $(ii)$ oriented and $(iii)$ of a compact closure. We moreover require $S$ to be such that every edge $e\in{\cal E}$ can be {\em adapted} to $S$ in the following sense \cite{area}: $e$ can be divided into a finite number of edges $\{e_1,\ldots,e_N\}$, i.e.
\[
e=e_N\circ e_{N-1}\circ\ldots\circ e_2\circ e_1,
\] 
each of them either
\begin{enumerate}
\item is contained in the closure $\overline{S}$; 
\item has no common points with $S$; 
\item has exactly one common point with $S$ being one of its two distinct endpoints.    
\end{enumerate}
We will call such a submanifold a {\em face}. A set of all faces in $\Sigma$ will be denoted by $\cal S$.   

A three-dimensional submanifold $V$ of $\Sigma$ of a compact closure and of an orientation inherited from $\Sigma$ will be called a {\em region}. A set of all regions in $\Sigma$ will be denoted by $\cal V$.

\subsection{Elementary degrees of freedom}

Note that the variables $(\xi^I,\theta^J)$ and $(\zeta_K,r_L)$ parameterizing the phase space $P\times\Theta$ are respectively, zero-forms (functions), one-forms, three-forms and two-forms which can be naturally integrated over submanifolds of $\Sigma$ of appropriate dimensions.     

Thus every point $y\in\Sigma$ defines naturally a function on $\Theta$:
\begin{equation}
\Theta\ni \theta\mapsto \kappa^I_y(\theta):=\xi^I(y)\in\R.
\label{k-y}
\end{equation}  
Similarly, every edge $e$ defines a function on $\Theta$: 
\begin{equation}
\Theta\ni \theta\mapsto \kappa^J_e(\theta):=\int_e\theta^J\in\R.
\label{k-e}
\end{equation}
We choose the set $\K$ of configurational elementary d.o.f  as follows 
\[
\K:=\{\ \kappa^I_y, \kappa^J_e \ | \ I,J=1,2,3, \, y\in\Sigma, \, e\in{\cal E} \ \}.
\]
It is easy to realize that the functions in $\K$ separate points in $\Theta$. 

Note that for every $I=1,2,3$, every $e\in{\cal E}$ and every pair of edges $e_1,e_2\in{\cal E}$ for which the composition  $e_2\circ e_1$ makes sense    
\begin{align}
&\kappa^I_{e^{-1}}=-\kappa^I_{e}, && \kappa^I_{e_2\circ e_1}=\kappa^I_{e_2}+\kappa^I_{e_1}.
\label{keke}
\end{align}
  
Every region $V$ defines a function on $P$:
\begin{equation}
P\ni p\mapsto \varphi^V_{I}(p):=\int_V\zeta_I\in\R.
\label{phi-V}
\end{equation} 
Similarly, every face $S$ defines a function on $P$:
\begin{equation}
P\ni p\mapsto \varphi^S_{J}(p):=\int_S r_J\in\R.
\label{phi-S}
\end{equation}      
We choose the set $\F$ of momentum elementary d.o.f  as follows
\[
\F:=\{ \ \varphi^V_{I}, \varphi^S_{J}\ | \ I,J=1,2,3, \, V\in{\cal V}, \, S\in{\cal S} \ \}.
\]
It is not difficult to check that the functions in $\F$ separate points in $P$. 

\subsection{Finite sets of configurational elementary d.o.f. \label{fin-sets}}

Let $K=\{\kappa_{1},\ldots,\kappa_N\}\subset{\cal K}$ be a finite set of elementary d.o.f.. We say that $\theta\in \Theta$ is $K$-related to $\theta'\in \Theta$,
\[
\theta\sim_{K} \theta',
\]
if for every $\kappa_{\alpha}\in {K}$      
\[
\kappa_{\alpha}(\theta)=\kappa_{\alpha}(\theta').
\]
Clearly, the relation $\sim_{K}$ is an equivalence one. Therefore it defines a quotient space
\begin{equation}
\Theta_{K}:=\Theta/\sim_{K}.
\label{quot}
\end{equation}
Note now that there exist $(i)$ a canonical projection from $\Theta$ onto $\Theta_K$:  
\begin{equation}
\Theta\ni \theta\mapsto\pr_{K}(\theta)=[\theta]\in \Theta_{K}
\label{pr-K}
\end{equation}
and $(ii)$ an {\em injective} map\footnote{Note that each set $K$ is unordered, thus to define the map $\tilde{K}$ one has to order elements of $K$. However, every choice of the ordering is equally well suited for our purposes and nothing essential depends on the choice. Therefore we will neglect this subtlety in what follows.} from $\Theta_{K}$ into $\R^N$:
\begin{equation}
\Theta_{K}\ni[\theta]\mapsto\tilde{K}([\theta]):=(\kappa_{1}(\theta),\ldots,\kappa_{N}(\theta))\in\R^N,
\label{k-inj}
\end{equation}
where $N$ is the number of elementary d.o.f. constituting $K$ and $[\theta]$ denotes the equivalence class of $\theta$ defined by the relation $\sim_{K}$.       

We will say that elementary d.o.f. in ${K}=\{\kappa_{1},\ldots,\kappa_{N}\}$ are {\em independent} if the image of $\tilde{K}$ is an $N$-dimensional submanifold of $\R^N$. A quotient space $\Theta_K$ given by a set $K$ of independent d.o.f. will be called a {\em reduced configuration space}.

\begin{lm}
Let $u=\{y_1,\ldots,y_M\}$ be a finite collection of points in $\Sigma$ and $\gamma=\{e_1,\ldots,e_N\}$ be a graph such that either $u$ or $\gamma$ is not an empty set ($N,M\geq0$ but $N+M>0$). Then for every $(z^I_{i},x^J_{j})\in\R^{3M}\times \R^{3N}$ there exists $\theta\in\Theta$ such that  
\begin{align*}
&\kappa^I_{y_{i}}(\theta)=z^I_{i},&&\kappa^J_{e_{j}}(\theta)=x^J_{j}
\end{align*}
for every $I,J=1,2,3$, $i=1,\ldots,M$ and $j=1,2,\ldots,N$.
\label{lm-Kug-xi}
\end{lm}
\noindent This lemma proven in \cite{q-suit} guarantees that if  
\begin{equation}
{K}_{u,\gamma}:=\{ \ {\kappa}^I_{y_1},\ldots,{\kappa}^I_{y_M},\kappa^J_{e_1},\ldots,\kappa^J_{e_N} \ | \ I,J=1,2,3 \ \}.
\label{K-ug}
\end{equation}
then
\begin{equation}
\Theta_{K_{u,\gamma}}\cong\R^{3M}\times\R^{3N},
\label{TKug-RN}
\end{equation}
under the map $\tilde{K}_{u,\gamma}$, i.e. $\tilde{K}_{u,\gamma}$ is a bijection. It means in particular that the d.o.f. constituting $K_{u,\gamma}$ are independent and $\Theta_{K_{u,\gamma}}$ is a reduced configuration space. We are also allowed to conclude that if $K$ is a one-element subset of $\K$ then $\tilde{K}$ is a bijection and consequently $K$ is a set of independent d.o.f. and $\Theta_K$ is a reduced configuration space.

Consider now a finite set $K$ of configurational elementary d.o.f. containing some (possibly none) d.o.f. \eqref{k-y} and some (possibly none) d.o.f. \eqref{k-e}. Let $u$ be a set of points defining elements of $K$ of the type \eqref{k-y} and let $E$ be a set of edges defining elements of $K$ of the type \eqref{k-e}. Let $\gamma$ be a graph related to $E$ as stated in Lemma \ref{E-gamma}. Since every $e\in E$ is a combination of edges of $\gamma$ and their inverses we can apply Equations \eqref{keke} to each $\kappa^I_e\in K$ to conclude that $\kappa^I_e$ is a linear combination of d.o.f. in $\K_{u,\gamma}$.  

\begin{cor}
For every finite set $K$ of configurational elementary d.o.f. there exists a finite set $u$ of points of $\Sigma$ and a graph $\gamma$ such that every d.o.f. in $K$ is a linear combination of d.o.f. in $K_{u,\gamma}$.
\label{K-lc-Kug}   
\end{cor}

Note now that if $\Theta_K$ is a reduced configuration space then the map $\tilde{K}$ can be used to define a differential structure on the space. It may happen that a set $K'$ of independent d.o.f. distinct from $K$ defines the same space: $\Theta_K=\Theta_{K'}$ i.e. $[\theta]=[\theta]'$ for every $\theta\in\Theta$, where $[\theta]'$ denotes the equivalence class of $\theta$ defined by the relation $\sim_{K'}$. Assume that then the differential structures on $\Theta_K=\Theta_{K'}$ given by $\tilde{K}$ and $\tilde{K}'$ coincide (we will prove soon that this {\em is} the case). Then following \cite{al-hoop} we can introduce the notion of cylindrical functions:
\begin{df}
We say that a function $\Psi:\Theta\to\C$ is a cylindrical function compatible with the set ${K}$ of independent d.o.f. if 
\begin{equation}
\Psi=\pr^*_{K}\,\psi
\label{Psi-cyl}
\end{equation}
for some smooth function $\psi:\Theta_{K}\to\C$. 
\end{df}
\noindent Note that each configurational elementary d.o.f. $\kappa$ is a cylindrical function compatible with $K=\{\kappa\}$. 

Denote by $\Cyl$ a complex linear space spanned by all cylindrical functions on $\Theta$.  

Let $\mathbf{K}$ be a set of all sets of independent d.o.f.. There holds the following important proposition \cite{q-stat}:
\begin{pro}
Suppose that there exists a subset $\mathbf{K}'$ of $\mathbf{K}$ such that  for every finite set $K_0$ of configurational elementary d.o.f.  there exists $K'_0\in\mathbf{K}'$ satisfying the following conditions: 
\begin{enumerate}
\item the map $\tilde{K}'_0$ is a bijection; 
\item each d.o.f. in $K_0$ is a linear combination of d.o.f. in $K'_0$.  
\end{enumerate} 
Then
\begin{enumerate}
\item for every set $K\in\mathbf{K}$ the map $\tilde{K}$ is a bijection. Consequently, $\Theta_K\cong \R^N$ with $N$ being the number of elements of $K$ and the map $\tilde{K}$ defines a linear structure on $\Theta_K$ being the pull-back of the linear structure on $\R^N$; if $\Theta_{K}=\Theta_{K'}$ for some other set $K'\in\mathbf{K}$ then the linear structures defined on the space by $\tilde{K}$ and $\tilde{K}'$ coincide.
\item if a cylindrical function $\Psi$ compatible with a set $K\in\mathbf{K}$ can be expressed as
\[
\Psi=\pr_{K'}\psi',
\]  
where $K'\in\mathbf{K}$ and $\psi'$ is a complex function on $\Theta_{K'}$ then $\psi'$ is smooth and consequently $\Psi$ is compatible with $K'$;     
\item for every element $\Psi\in\Cyl$ there exists a set $K\in\mathbf{K}'$ such that $\Psi$ is compatible with $K$.    
\end{enumerate}
\label{big-pro}
\end{pro}
\noindent It follows from Lemmas \ref{lm-Kug-xi} and Corollary \ref{K-lc-Kug} that a subset of $\mathbf{K}$ consisting of all sets $K_{u,\gamma}$, where $u$ runs through all finite subsets of $\Sigma$ and $\gamma$ runs through all graphs in $\Sigma$ satisfies the requirement imposed on the set $\mathbf{K}'$ by the proposition. Thus, according to Assertion 1 of the proposition, on every reduced configuration space $\Theta_K$ there exists a natural {\em linear structure} and, consequently, a natural differential structure. This means that the space $\Cyl$ introduced just above the proposition is well defined, Assertions 2 holds and by virtue of Assertion 3 for every element $\Psi\in \Cyl$ there exist a finite set $u$ of points in $\Sigma$ and a graph $\gamma$ such that $\Psi$ is compatible with $K_{u,\gamma}$.      

A simple but useful consequence of the results above is that on every reduced configuration space $\Theta_K$, where $K=\{\kappa_1,\ldots,\kappa_N\}$,  one can define a linear coordinate frame $(x_1,\ldots,x_N)$: 
\begin{equation}
\Theta\ni [\theta]\mapsto x_{\alpha}([\theta]):=\kappa_{\alpha}(\theta)\in \R, 
\label{lin-coor-0}
\end{equation}
in other words,
\begin{equation}
(\,x_1[\theta],\ldots,x_N([\theta])\,)=\tilde{K}([\theta]).
\label{x-tilK}
\end{equation}
The frame \eqref{lin-coor-0} will be called {\em natural coordinate frame on} $\Theta_K$. 

\subsection{Operators corresponding to momentum d.o.f.}

Consider a finite collection $u=\{y_1,\ldots,y_M\}$ of points in $\Sigma$ and a graph $\gamma=\{e_1,\ldots,e_N\}$ such that either $u$ or $\gamma$ is not an empty set ($N,M\geq0$ but $N+M>0$). Let us introduce a special notation for natural coordinates \eqref{lin-coor-0} defined on a reduced configuration space $\Theta_{K_{u,\gamma}}$: we will denote the coordinates by $(z^I_{i},x^J_{j})$, $I,J=1,2,3$, $i=1,\ldots,M$ if $M>0$ and $j=1,\ldots,N$ if $N>0$, where
\begin{align}
&z^I_{i}([\theta]):=\kappa^I_{y_{i}}(\theta), && x^J_{j}([\theta]):=\kappa^J_{e_{j}}(\theta),
\label{lin-coor}
\end{align}      
(here $[\theta]\in\Theta_{K_{u,\gamma}}$). The coordinates define vector fields 
\[
\{\ \partial_{z^I_{i}}, \partial_{x^J_{j}}\ \}
\]
on $\Theta_{K_{u,\gamma}}$---these vector fields will be used to express operators  defined on $\Cyl$ by the momentum d.o.f. \eqref{phi-V} and \eqref{phi-S}.
 
\subsubsection{Operators corresponding to d.o.f. \eqref{phi-V}}

Using the Poisson bracket \eqref{poiss} we define an operator
\begin{equation}
\Cyl\ni\Psi\mapsto\hat{\varphi}^V_I\Psi:=\{\varphi^V_I,\Psi\}\in\Cyl.
\label{hat-zeta-0}
\end{equation}
We know already (see the discussion just below Proposition \ref{big-pro}) that $\Psi$ is compatible with a set $K_{u,\gamma}$, i.e. $\Psi=\pr^*_{K_{u,\gamma}}\psi$ for a function $\psi$ defined on $\Theta_{K_{u,\gamma}}$. Assume that $u=\{y_1,\ldots,y_M\}$. Then  
\begin{equation}
\hat{\varphi}^V_I\Psi=\sum_{L=1}^3\sum_{l=1}^M \pr^*_{K_{u,\gamma}}(\partial_{z^L_{l}}\psi)\{\varphi^V_I,\kappa^L_{y_{l}}\},
\label{hat-zeta-1}
\end{equation}
where $\{\partial_{z^L_{l}}\}$ are vector fields on $\Theta_{K_{u,\gamma}}$ defined by the natural coordinates \eqref{lin-coor}. To find an explicite expression for $\{\varphi^V_I,\kappa^L_{y_{l}}\}$ note that
\begin{align*}
&\varphi^V_I(p)=\int_\Sigma {\cal I}_V \zeta_I, && \kappa^L_y(\theta)=\int_\Sigma\delta_y\xi^L,
\end{align*}
 where ${\cal I}_V$ is the characteristic function of the region $V$ and $\delta_y$ is the Dirac distribution supported at $y\in\Sigma$. Hence
\[
\{\varphi^V_I,\kappa^L_{y}\}=-\delta^L{}_I\int_\Sigma {\cal I}_V\delta_y.
\]
Let
\begin{equation}
\varepsilon(V,y):=
\begin{cases}
-1& \text{if $y\in V$}\\
0 & \text{otherwise}
\end{cases}.
\label{Vy}
\end{equation}
Then
\begin{equation}
\{\varphi^V_I,\kappa^L_{y}\}=\hat{\varphi}^V_I\kappa^L_{y}=\delta^L{}_I\,\varepsilon(V,y).
\label{ze-xi-const}
\end{equation}
Let us emphasize that  $\hat{\varphi}^V_I\kappa^L_{y}$ is a {\em constant} real cylindrical functions which since now will be treated as a real number. Thus finally
\begin{equation}
\hat{\varphi}^V_I\Psi=\sum_{l=1}^N \varepsilon(V,y_{l})\,\pr^*_{K_{u,\gamma}}(\partial_{z^I_{l}}\psi),
\label{hat-zeta}
\end{equation}
which means that $\hat{\varphi}^V_I$ preserves the space $\Cyl$. Thus $\hat{\varphi}^V_I$ is a linear operator on the space.
  
\subsubsection{Operators corresponding to d.o.f. \eqref{phi-S}}
    
With every elementary d.o.f. $\varphi^S_J\in\F$ we associate a flux operator $\hat{\varphi}^S_J$ \cite{acz}---it is a linear operator on $\Cyl$ defined via a suitable regularization of the Poisson bracket $\{{\varphi}^S_J,\Psi\}$, where $\Psi\in\Cyl$. Again, we express the cylindrical function as $\Psi=\pr^*_{K_{u,\gamma}}\psi$ for some set $K_{u,\gamma}$ and a function $\psi$ on $\Theta_{K_{u,\gamma}}$. Assume that $\gamma=\{e_1,\ldots,e_N\}$. Then the operator $\hat{\varphi}^J_S$ acts on $\Psi$ as follows:   
\begin{equation}
\hat{\varphi}^S_J\Psi:=\sum_{j=1}^N\eps(S,e_{j})\, \pr_{K_{u,\gamma}}^*(\partial_{x^J_{j}}\psi)\in\Cyl,
\label{hphi_S}
\end{equation}
where $\{\partial_{x^J_{j}}\}$ are vector fields on $\Theta_{K_{u,\gamma}}$ given by the coordinate frame \eqref{lin-coor} and  each $\eps(S,e_{j})$ is a certain real number.

To define the number $\eps(S,e_j)$ we adapt the edge $e_{j}$ to $S$ obtaining thereby a set of edges $\{e_{j1},\ldots,e_{jn}\}$ and define a function $\eps$ on this set: $\eps(e_{ja})=0$ if $e_{ja}$ is contained in $\bar{S}$ or has no common points with $S$; in remaining cases           
\begin{enumerate}
\item $\eps(e_{ja}):=\frac{1}{2}$ if $e_{ja}$ is either 'outgoing' from $S$ and placed 'below' the face  or is 'incoming' to $S$ and placed 'above' the face;
\item $\eps(e_{ja}):=-\frac{1}{2}$ if $e_{ja}$ is either 'outgoing' from $S$ and placed 'above' the face  or is 'incoming' to $S$ and placed 'below' the face. 
\end{enumerate}
Here the terms 'outgoing' and 'ingoing' refer to the orientation of the edges (which is inherited from the orientation of $e_{j}$) while the terms 'below' and 'above' refer to the orientation of the normal bundle of $S$ defined naturally by the orientations of $S$ and $\Sigma$. Then we define
\[
\eps(S,e_{j}):=\sum_{a=1}^n \eps(e_{ja}).
\]       

It is not difficult to realize that for every edge $e\in{\cal E}$ 
\begin{equation}
\hat{\varphi}^S_J\kappa^L_{e}=\delta^L{}_J\,\eps(S,e)
\label{hphiS-ke}
\end{equation}
which means that $\hat{\varphi}^S_J\kappa^L_{e}$ is a {\em constant} real cylindrical function which since now will be regarded as a real number. 

\subsubsection{Linear space of the operators}

Let us introduce a space $\hat{\F}$ as a real linear space spanned by all operators \eqref{hat-zeta-0} and \eqref{hphi_S}:  
\[
\hat{\cal F}:=\spn_{\R}\{\ \hat{\varphi}\ | \ \varphi\in\F\ \}.
\] 
Thus an element $\hat{\varphi}$ of $\hat{\F}$ is of the following form 
\[
\hat{\varphi}=\sum_{Ii} A^I_{i}\hat{\varphi}^{V_{i}}_I+\sum_{Jj} B^J_{j}\hat{\varphi}^{S_{j}}_J,
\]
where $A^I_{i},B^J_{j}$ are real numbers and both sums are finite.

Let $\Psi=\pr^*_{K_{u,\gamma}}\psi$ be a cylindrical function compatible with $K_{u,\gamma}$. Then 
\begin{multline}
\hat{\varphi}\Psi=\pr^*_{K_{u,\gamma}}\Big(\sum_{Iil} A^I_{i}\,\varepsilon(V_{i},y_{l})\,\partial_{z^I_{l}}\psi+\sum_{Jjn} B^J_{j}\,\epsilon(S_{j},e_{n})\,\partial_{x^J_{n}}\psi\Big)=\\=\sum_{I l} \Big(\pr^*_{K_{u,\gamma}}\partial_{z^I_{l}}\psi\Big)\,\hat{\varphi}\kappa^I_{y_{l}}+\sum_{Jn} \Big(\pr^*_{K_{u,\gamma}}\partial_{x^J_{n}}\psi\Big)\,\hat{\varphi}\kappa^J_{e_{n}},
\label{hphi-Psi}
\end{multline}
where in the first step we used \eqref{hat-zeta} and \eqref{hphi_S} and in the second one we applied \eqref{ze-xi-const} and \eqref{hphiS-ke}.
 
\subsection{A directed set $(\Lambda,\geq)$}

All considerations above were preparatory steps to the crucial one which is a choice of a directed set $(\Lambda,\geq)$---once such a set is {\em chosen properly} the prescription described in \cite{q-stat} can be used to build from it a unique space of quantum states. 

\subsubsection{General assumptions imposed on $(\Lambda,\geq)$ \label{ad-as}}

Recall that $\mathbf{K}$ denotes a set of all sets of independent d.o.f.. Let $\hat{\mathbf{F}}$ be a set of all finite dimensional linear subspaces of $\hat{\F}$. A directed set $(\Lambda,\geq)$, where $\Lambda\subset\hat{\mathbf{F}}\times \mathbf{K}$,  is chosen properly if it satisfies the following {\bf Assumptions} \cite{q-stat}:
\begin{enumerate}
\item 
\begin{enumerate}
\item for each finite set $K_0$ of configurational elementary d.o.f.  there exists $(\hat{F},K)\in\Lambda$ such that each $\kappa\in K_0$ is a cylindrical function compatible with $K$; \label{k-Lambda}
\item for each finite set $F_0$ of momentum elementary d.o.f. there exists $(\hat{F},K)\in\Lambda$ such that $\hat{\varphi}\in\hat{F}$ for every $\varphi\in F_0$; \label{f-Lambda}
\end{enumerate}
\item \label{RN} 
if $(\hat{F},K)\in\Lambda$ then the image of the map $\tilde{K}$ given by \eqref{k-inj} is $\R^N$ (where $N$ is the number of elements of $K$)---in other words, $\tilde{K}$ is a bijection and consequently  
\[
\Theta_K\cong\R^N.
\] 
\item 
if $(\hat{F},K)\in\Lambda$, then 
\begin{enumerate}
\item for every $\hat{\varphi}\in \hat{\F}$ and for every cylindrical function $\Psi=\pr_K^*\psi$ compatible with $K=\{\kappa_1,\ldots,\kappa_N\}$ 
\[
\hat{\varphi}\Psi=\sum_{\alpha=1}^N\Big(\pr^*_K\partial_{x_{\alpha}}\psi\Big)\hat{\varphi}\kappa_{\alpha},
\]   
where $\{\partial_{x_{\alpha}}\}$ are vector fields on $\Theta_K$ given by the natural coordinate frame \eqref{lin-coor-0}; \label{comp-f} 
\item for every $\hat{\varphi}\in \hat{\F}$ and for every $\kappa\in K$ the cylindrical function $\hat{\varphi}\kappa$ is a real {\em constant} function on $\Theta$; \label{const}
\end{enumerate}
\item if $(\hat{F},K)\in\Lambda$ and $K=\{\kappa_{1},\ldots,\kappa_{N}\}$ then $\dim\hat{F}=N$; moreover, if $(\hat{\varphi}_1,\ldots,\hat{\varphi}_N)$ is a basis of $\hat{F}$ then an $N\times N$ matrix $G=(G_{\beta\alpha})$ of components
\[
G_{\beta\alpha}:=\hat{\varphi}_{\beta}\kappa_{\alpha}
\]    
is {\em non-degenerate}. \label{non-deg}
\item  if $(\hat{F},K'),(\hat{F},K)\in\Lambda$ and $\Theta_{K'}=\Theta_{K}$ then  $(\hat{F},K')\geq(\hat{F},K)$; \label{Q'=Q} 
\item if $(\hat{F}',K')\geq(\hat{F},K)$ then 
\begin{enumerate}
\item each d.o.f. $K$ is {\em a linear combination} of d.o.f. in $K'$; \label{lin-comb}
\item $\hat{F}\subset\hat{F}'$. \label{FF'}
\end{enumerate} 
\end{enumerate} 

\subsubsection{Speckled graphs \label{speckl}}

In the considerations above an important role was played by sets $\{K_{u,\gamma}\}$. Therefore one may try to use these sets to define a set $\Lambda$ as one consisting of pairs $(\hat{F},K_{u,\gamma})$ which satisfy all Assumptions listed in the previous section. However, we will not use all sets $K_{u,\gamma}$ to define $\Lambda$ but will restrict ourselves to some of them. 

To justify our decision let us refer to the general construction presented in \cite{q-stat} (see its outline in Section \ref{out}). According to it for every $(\hat{F},K_{u,\gamma})\equiv\lambda\in\Lambda$ we will have to associate a Hilbert space $\h_\lambda$ of some square integrable functions on $\Theta_{K_{u,\gamma}}$ (then density operators on all such Hilbert spaces will be used to build the space $\D$). It seems to us that it would be highly desirable if one could define on each Hilbert space $\h_\lambda$ a sort of quantum geometry related to the Riemannian geometry of $\Sigma$. But to achieve this we have to guarantee that the d.o.f. in $K_{u,\gamma}$ can be used to extract some consistent information about the geometry. Since the geometry is given by the metric $q$ we have to require that the d.o.f. provide some consistent information about the metric. Let us now analyze this issue more carefully.

The metric $q$ is defined by \eqref{q} in terms of the variables $(\theta^0,\theta^J)$,
\begin{equation}
q=-\theta^0\ot\theta^0+\delta_{IJ}\theta^I\ot\theta^J,
\label{q-0J}
\end{equation}
Thus we should require the d.o.f. constituting $K_{u,\gamma}$ to give us information about both the one-form $\theta^0$ and the other forms $(\theta^J)$. Of course, information about $(\theta^J)$ is given by d.o.f. $\{\kappa^J_e\}$ defined as integrals \eqref{k-e} of the forms along edges of the graph $\gamma$. Therefore, to achieve a consistency, we should be able to approximate integrals of $\theta^0$ along the edges by means of the d.o.f..       

Consider $K_{u,\gamma}$ defined by a set $u=\{y\}$ and a graph $\gamma=\{e\}$. Thus $K_{u,\gamma}=\{\kappa^I_y,\kappa^J_e\}$. Suppose now that $y\in e$. Because 
\[
\theta^0=\sgn(\theta^J)\frac{\xi_I}{\sqrt{1+\xi_K\xi^K}}\,\theta^I
\]
(see \eqref{old-new}) the integral 
\begin{equation}
\int_e\theta^0
\label{e-th0}
\end{equation}
can be approximated {\em modulo the factor} $\sgn(\theta^I)$ by
\begin{equation}
\frac{\xi_I(y)}{\sqrt{1+\xi_L(y)\xi^L(y)}}\int_e\theta^I=\frac{\kappa_{Iy}(\theta)\kappa^I_e(\theta)}{\sqrt{1+\kappa_{Ly}(\theta)\kappa^L_y(\theta)}},
\label{t0-app}
\end{equation}
where $\theta=(\theta^0,\theta^J)=(\xi^I,\theta^J)\in\Theta$. If, however, $y\not\in e$ then in general we cannot obtain from the d.o.f. $\{\kappa^I_{y},\kappa^J_e\}$ a good approximation of the integral \eqref{e-th0}. 

Thus we conclude that to define the set $\Lambda$ we should use sets $\{K_{u,\gamma}\}$ such that each point $y\in u$ belongs to an edge of $\gamma$ and for each edge $e$ of $\gamma$ the intersection $e\cap u$ consists of exactly one point. 

However, this conclusion may seem to be a bit premature because while drawing it we neglected the lack of the factor $\sgn(\theta^I)$ in the formula \eqref{t0-app}. It turns out that, given $K_{u,\gamma}$, the d.o.f. in it do not contain any information about the factor---this fact is a consequence of the following lemma \cite{q-suit}:
\begin{lm}
Let $\gamma=\{e_1,\ldots,e_N\}$ be a graph. Then for every $(x^I_{i})\in\R^{3N}$ there exists a global coframe $(\theta^I)$ on $\Sigma$ compatible (incompatible) with the orientation of the manifold such that  
\[
\int_{e_{i}}\theta^I=x^I_{i}
\]      
for every $I=1,2,3$ and $i=1,2,\ldots,N$. 
\label{theta-x3}
\end{lm}
\noindent The lemma means that for every $\theta\equiv(\xi^I,\theta^J)\in\Theta$ the equivalence class $[\theta]$ defined by $K_{u,\gamma}$ contains points of $\Theta$ given by global coframes compatible as well as global coframes incompatible with the orientation of $\Sigma$. Hence no function on $K_{u,\gamma}$ can be an approximation of $\sgn(\theta^I)$.  

Is the impossibility to approximate $\sgn(\theta^I)$ by functions on $K_{u,\gamma}$ a problem? It could be if relevant quantities like ones describing the geometry of $\Sigma$ as well as  constraints and Hamiltonians depended on $\sgn(\theta^I)$. 

Note that the metric $q$ is a quadratic function of $\theta^0$. Therefore the sign $\sgn(\theta^I)$ is irrelevant for the metric---see the formula \eqref{q-xi} expressing the metric in terms of $(\xi^I,\theta^J)$. Consequently, geometric quantities on $\Sigma$ including the Hodge operator $*$ given by $q$ and the orientation of $\Sigma$ do not depend on $\sgn(\theta^I)$. 

Regarding the constraints (and the Hamiltonians\footnote{The Hamiltonians of TEGR and YMTM are sums of constraints.}) of TEGR and YMTM, an important observation is that they are {\em quite specific} functions of $(\theta^A,p_B)$ and a variable $\xi^A$ defined as a solution of the following equation system\footnote{Clearly, $\xi^A$ is a function of $\theta^A$---see \cite{os}.} \cite{nester}:
\begin{align*}
&\xi_A\theta^A=0&& \xi_A\xi^A=-1.
\end{align*}
Namely, these three variables appear in the constraints exclusively in the form of either $(i)$ a contraction with respect to the index $A$ e.g. $\xi^Adp_A$ or $(ii)$ scalar products defined by $\eta$ (or its inverse) e.g. $\eta_{AB}d\theta^A\we*d\theta^B$ (or $\eta^{AB}p_A\we*p_B)$. Since the matrix $(\eta_{AB})$ (and its inverse) is diagonal two time-like components (that is, components with $A=0$) of the variables always multiply each other e.g. 
\[
\xi^Adp_A=\xi^0dp_0+\xi^Idp_I\quad\quad\text{or} \quad\quad \eta_{AB}d\theta^A\we*d\theta^B=-d\theta^0\we*d\theta^0+d\theta^I\we*d\theta^I.
\]   
On the other hand every time-like component of the variables under consideration is proportional to $\sgn(\theta^I)$ \cite{ham-nv}: 
\begin{align*}
&p_0=\sgn(\theta^L)\sqrt{1+\xi_J\xi^J}\,\vec{\theta}^I\lr\zeta_I,&&\theta^0=\sgn(\theta^J)\frac{\xi_I}{\sqrt{1+\xi_L\xi^L}}\,\theta^I, && \xi^0=\sgn(\theta^I)\sqrt{1+\xi_J\xi^J},
\end{align*}
and the space-like components (that is, components with $A\in\{1,2,3\}$) of the variables are independent of the factor. Thus the factor appears in the constraints exclusively as $[\sgn(\theta^I)]^2\equiv 1$ and, consequently, the constraints expressed in terms of $(\zeta_I,r_J,\xi^K,\theta^L)$ are independent of $\sgn(\theta^I)$---see \cite{ham-nv} for explicite expressions of them.

Thus (at least at this stage) the impossibility to express $\sgn(\theta^I)$ by functions on $K_{u,\gamma}$ does not seem to cause any problem. 

Let us turn back to the conclusion placed just below the formula \eqref{t0-app}. It motivates us to introduce a special kind of graphs:
\begin{df}
A speckled graph $\dot{\gamma}$ in $\Sigma$ is a pair $(u,\gamma)$, where $u$ is a finite set of points in $\Sigma$ and $\gamma$ is a graph such that there exists a surjective map $\chi:\gamma\to u$ such that $\chi(e)\in e$ for every $e\in\gamma$.         
\end{df}

Let $\dot{\gamma}=(u,\gamma)$ be a speckled graph. We will denote
\[
K_{u,\gamma}\equiv K_{\dot{\gamma}}.
\] 

Now the conclusion mentioned above can be reformulated: to define a set $\Lambda$ for TEGR we should use sets $\{K_{\dot{\gamma}}\}$ given by all speckled graphs in $\Sigma$. Let us now take a closer look at these graphs.   

\subsubsection{Properties of speckled graphs}

Consider now a pair $(u,\gamma)$, where $u$ is a finite set of points and $\gamma$ is a graph. Of course, $(u,\gamma)$ may not be a speckled graph because it may happen that $(i)$ there exist elements of $u$ which do not belong to any edge of $\gamma$; $(ii)$ there are edges of $\gamma$ such that no point in $u$ belongs to the edges or $(iii)$ there are edges of $\gamma$ such that  two or more distinct points of $u$ belong to each of the edges. Note however that $(u,\gamma)$ can be easily transformed to a speckled graph as follows: in a case of a point $y$ in $u$ of the sort $(i)$ one can choose an edge $e$ such that $y$ is the only point of $u$ which belongs to $e$ and $\gamma\cup \{e\}$ is a graph; in a case of an edge of the sort $(ii)$ one can single out a point in $e$ and add it to $u$; in a case of an edge of the sort $(iii)$ one can divide it into smaller edges such that each smaller one contains exactly one point belonging to $u$.            

\begin{cor}
 For every pair $(u,\gamma)$, where $u$ is a finite set of points and $\gamma$ is a graph, there exists a speckled graph $\dot{\gamma}'=(u',\gamma')$ such that $u'\supset u$ and $\gamma'\geq\gamma$.
\label{ug-dg}     
\end{cor}

We will say that a speckled graph $\dot{\gamma}'$ is {\em greater or equal} to a speckled graph $\dot{\gamma}$,
\begin{equation}
\dot{\gamma}'=(u',\gamma')\geq\dot{\gamma}=(u,\gamma),
\label{spg->}
\end{equation}
if $u'\supset u$ and $\gamma'\geq\gamma$. 

\begin{lm}
The set of all speckled graphs in $\Sigma$ equipped with the relation \eqref{spg->} is a directed set. 
\end{lm}
\begin{proof}
The relation \eqref{spg->} is obviously transitive. Let us then show that for every two speckled graphs $\dot{\gamma}'=(u',\gamma'),\dot{\gamma}=(u,\gamma)$ there exists a speckled graph $\dot{\gamma}''$ such that $\dot{\gamma}''\geq\dot{\gamma}'$ and $\dot{\gamma}''\geq\dot{\gamma}$.

Let $u_0:=u'\cup u$ and let $\gamma_0$ be a graph such that $\gamma_0\geq\gamma',\gamma$. Due to Corollary \ref{ug-dg} there exists a speckled graph $\dot{\gamma}''=(u'',\gamma'')$ such that $u''\supset u_0$ and $\gamma''\geq\gamma_0$. Thus $\dot{\gamma}''$ is the desired speckled graph.   
\end{proof}

\begin{lm}
For every finite set $K$ of configurational elementary d.o.f. there exists a speckled graph ${\dot{\gamma}}$ such that each d.o.f. in $K$ is a linear combination of d.o.f. in $K_{\dot{\gamma}}$.    
\label{K-Kdg}
\end{lm}
\begin{proof}
By virtue of Corollary \ref{K-lc-Kug} each d.o.f. in $K$ is a linear combination of d.o.f. in $K_{u,\gamma}$ given by a pair $(u,\gamma)$ consisting of a finite set $u$ of points in $\Sigma$ and a graph $\gamma$. On the other hand, Equations \eqref{keke} and Corollary \ref{ug-dg} allow us to conclude that there exists a speckled graph $\dot{\gamma}'$ such that each d.o.f. in $K_{u,\gamma}$ is a linear combination of d.o.f. in $K_{\dot{\gamma}'}$.    
\end{proof}

Due to Lemmas \ref{lm-Kug-xi} and \ref{K-Kdg} a set of all sets $K_{\dot{\gamma}}$, where $\dot{\gamma}$ runs through all speckled graphs in $\Sigma$, meets the requirement satisfied by the set $\mathbf{K}'$ in Proposition \ref{big-pro}. This means that for every $\Psi\in\Cyl$ there exists a speckled graph $\dot{\gamma}$ such that $\Psi$ is a cylindrical function compatible with $K_{\dot{\gamma}}$.         

\begin{lm}
$\dot{\gamma}'\geq\dot{\gamma}$ if and only if each d.o.f. in $K_{\dot{\gamma}}$  is linear combination of d.o.f. in $K_{\dot{\gamma}'}$. 
\label{g'g-lin}
\end{lm}
\begin{proof}
If $\dot{\gamma}'\geq\dot{\gamma}$ then using Equations \eqref{keke} we can easily conclude that each d.o.f. in $K_{\dot{\gamma}}$  is linear combination of d.o.f. in $K_{\dot{\gamma}'}$. 

Let $\dot{\gamma}'=(u',\gamma')$ and  $\dot{\gamma}=(u,\gamma)$, where $u=\{y_1,\ldots,y_{M}\}$. Suppose now that each d.o.f. in $K_{\dot{\gamma}}$  is a linear combination of d.o.f. in $K_{\dot{\gamma}'}$. Taking into account the formula \eqref{k-y} we see that then each $\kappa^I_{y_{i}}$ belonging to $\K_{\dot{\gamma}}$ belongs to $\K_{\dot{\gamma}'}$. Thus $u'\supset u$. Now let us show that $\gamma'\geq \gamma$. 

To this end consider a set $\Omega$ of one-forms on $\Sigma$ defined as follows: a one-form $\varpi$ belongs to $\Omega$ if there exists one-forms $\theta^2,\theta^3$ such that $(\varpi,\theta^2,\theta^3)$ form a global coframe on $\Sigma$. Then for any real functions $\{\xi^I\}$, $I=1,2,3$,  on $\Sigma$ the collection $\theta=(\xi^I,\varpi,\theta^2,\theta^3)$ is an element of $\Theta$. Given $e\in{\cal E}$, we define a real function $\kappa_e$ on $\Omega$ 
\[
\kappa_{e}(\varpi):=\kappa^1_{e}(\theta)=\int_e\varpi
\]        
and apply Lemma \ref{lm-Kug-xi} to conclude that for every graph $\gamma_0=\{e_1,\ldots,e_{N_0}\}$ and for each $(x_1,\ldots,x_{N_0})\in\R^{N_0}$ there exists $\varpi\in\Omega$ such that $\kappa_{e_{i}}(\varpi)=x_{i}$. 

Suppose now that each d.o.f. in $K_{\dot{\gamma}}$  is a linear combination of d.o.f. in $K_{\dot{\gamma}'}$, where $\dot{\gamma}=(u,\gamma)$ and $\dot{\gamma}'=(u',\gamma')$. Obviously, a combination describing $\kappa^1_e$ defined by an edge $e$ of $\gamma$ cannot contain d.o.f. $\{\kappa^J_{y'}\}$ given by points $\{y'\}=u'$. Thus
\[
\kappa^1_e=A^i\kappa^1_{e'_i}+B^i\kappa^2_{e'_i}+C^i\kappa^3_{e'_i},
\]
where $A^i,B^j,C^k$ are constant coefficients and $\gamma'=\{e'_1,\ldots,e'_{N'}\}$. Given $\theta=(\xi^I,\theta^J)\in\Theta$, consider a family $\{\theta_t=(\xi^I,\theta^1,t\theta^2,\theta^3)\}$ of elements of $\Theta$, where the number $t>0$. Differentiating with respect to $t$ both sides of the following equations   
\[
\kappa^1_e(\theta_t)=A^i\kappa^1_{e'_i}(\theta_t)+B^i\kappa^2_{e'_i}(\theta_t)+C^i\kappa^3_{e'_i}(\theta_t)
\]
we obtain
\[
0=B^i\kappa^2_{e'_i}(\theta),
\]
hence $B^i=0$ by virtue of Lemma \ref{lm-Kug-xi}. Similarly we show that $C^i=0$. We conclude that each $\kappa^1_{e}\in K_{\dot{\gamma}}$ is a linear combinations of d.o.f $\{\kappa^1_{e'_{j}}\}\subset K_{\dot{\gamma}'}$ only. Thus each function in $\{\kappa_{e_1},\ldots,\kappa_{e_N}\}$ associated with edges of the graph $\gamma$ is a linear combination of functions $\{\kappa_{e'_1},\ldots,\kappa_{e'_{N'}}\}$ associated with edges of $\gamma'$.

Now, to conclude that $\gamma'\geq\gamma$ it is enough to apply the following lemma \cite{q-stat}: 
\begin{lm}
Let $\Omega$ be a set of one-forms on $\Sigma$ such that for every graph $\gamma_0=\{e_1,\ldots,e_{N_0}\}$ and for each $(x_1,\ldots,x_{N_0})\in\R^{N_0}$ there exists $\varpi\in\Omega$ such that
\[
\kappa_{e_{i}}(\varpi)=x_{i}.
\]    
Then $\gamma'=\{e'_1,\ldots,e'_{N'}\}\geq\gamma=\{e_1,\ldots,e_N\}$ if and only if each function in $\{\kappa_{e_1},\ldots,\kappa_{e_N}\}$ is a linear combination of functions $\{\kappa_{e'_1},\ldots,\kappa_{e'_{N'}}\}$.
\end{lm}

Thus $\gamma'\geq\gamma$ and, taking into account the previous result $u'\supset u$, we see that $\dot{\gamma}'\geq\dot{\gamma}$.    
\end{proof}

\subsubsection{Choice of  a directed set $\Lambda$}

Consider an element $\hat{F}$ of $\hat{\mathbf{F}}$ and an element $K=\{\kappa_1,\ldots,\kappa_{N}\}$ of $\mathbf{K}$. We say that a pair $(\hat{F},K)$ is {\em non-degenerate} if $\dim\hat{F}=N$ and an $(N\times N)$-matrix $G=(G_{\beta\alpha})$ of components          
\begin{equation}
G_{\beta\alpha}:=\hat{\varphi}_{\beta}\kappa_{\alpha},
\label{matr-G}
\end{equation}
where $(\hat{\varphi}_1,\ldots,\hat{\varphi}_{N})$ is a basis of $\hat{F}$, is non-degenerate.

\begin{df}
The set $\Lambda$ is a set of all non-degenerate pairs $(\hat{F},K_{\dot{\gamma}})\in\hat{\mathbf{F}}\times \mathbf{K}$, where $\dot{\gamma}$ runs through all speckled graphs in $\Sigma$.    
\label{df-Lambda}
\end{df}  

\begin{lm}
For every speckled graph $\dot{\gamma}$ in $\Sigma$ there exists $\hat{F}\in\hat{\mathbf{F}}$ such that $(\hat{F},K_{\dot{\gamma}})\in\Lambda$.
\label{every-g}
\end{lm}
\begin{proof}
Let $\dot{\gamma}=(u,\gamma)$, where $u=\{y_1,\ldots,y_M\}$ and $\gamma=\{e_1,\ldots,e_N\}$ ($M\leq N$). There exist regions $\{V_1,\ldots,V_M\}$ such that $V_{j}\cap u=y_{j}$. Consequently,
\[
\varepsilon(V_{j},y_{i})=-\delta_{ji}
\]
and introducing multi-labels $\alpha=(i,I)$ and $\beta=(j,J)$ we can write
\[
G^1_{\beta\alpha}:=-\hat{\varphi}^{V_{j}}_J\kappa^I_{y_{i}}=\delta^{I}{}_J\delta_{ji}=\delta_{\beta\alpha}.
\]
     
The independence of edges $\{e_1,\ldots,e_N\}$ of the graph $\gamma$ imply that there exists a set $\{S_1,\ldots,S_N\}$ of faces such that $e_{i}\cap S_{j}$ is empty if $i\neq j$ and consists of exactly one point distinct from the endpoints of $e_{i}$ if $i=j$. The orientations of the faces can be chosen in such a way that
\[
\epsilon(S_{j},e_{i})=\delta_{ji}
\]
Using the multi-labels $\alpha=(i,I)$ and $\beta=(j,J)$ we can write  
\[
G^2_{\beta\alpha}=\hat{\varphi}^{S_{j}}_J\kappa^I_{e_{i}}=\delta^{I}{}_J\delta_{ji}=\delta_{\beta\alpha}.
\]
Since 
\[
-\hat{\varphi}^{V_{j}}_J\kappa^I_{e_{i}}=0=\hat{\varphi}^{S_{j}}_J\kappa^I_{y_{i}}
\] 
the matrix $G$ given by \eqref{matr-G} for $K_{\dot{\gamma}}$ and 
\[
F_0:=\{\ -\hat{\varphi}^{V_{i}}_I,\hat{\varphi}^{S_{j}}_J\ | \ I,J=1,2,3;\, i=1,\ldots,M;\, j=1,\ldots,N\ \}
\]
is of the following form
\[
G=
\begin{pmatrix}
G^1 & \mathbf{0}\\
\mathbf{0} & G^2
\end{pmatrix}=\mathbf{1}
\] 
and, being the unit $(M+N)\times(M+N)$ matrix, is obviously non-degenerate. Thus if
\[
\hat{F}={\rm span}_{\R}\,F_0
\]
then elements of $F_0$ constitute a basis of $\hat{F}$ and $(\hat{F},K_{\dot{\gamma}})\in\Lambda$. 
\end{proof}

Now let us define  a relation $\geq$ on $\Lambda$: 
\begin{df}
Let $(\hat{F}',K_{\dot{\gamma}'}),(\hat{F},K_{\dot{\gamma}})\in\Lambda$. Then $(\hat{F}',K_{\dot{\gamma}'})\geq (\hat{F},K_{\dot{\gamma}})$ if and only if
\begin{align*}
&\hat{F}'\supset\hat{F} && \text{and} && \dot{\gamma}'\geq\dot{\gamma}.         
\end{align*}
\label{df-Lambda->}
\end{df}

\begin{lm}
$(\Lambda,\geq)$ is a directed set.
\label{Lambda-dir}
\end{lm}

Regarding a proof of the lemma it would be perhaps enough to refer to the proof of an analogous lemma in \cite{q-stat} concerning a set $\Lambda$ constructed for DPG saying that a proof of Lemma \ref{Lambda-dir} is a modification of the proof of that lemma in \cite{q-stat}. But yet taking into account the importance of Lemma \ref{Lambda-dir} to avoid any doubt we decided to present the proof explicitely. 

Before we will prove the lemma let us state some facts which will be used in the proof. Let $\bld{\Psi}$ be a subset of $\Cyl$. Then operators in $\hat{\F}$ restricted to $\bld{\Psi}$ are maps from $\bld{\Psi}$ into $\Cyl$. Since both $\Cyl$ and $\hat{\F}$ are linear spaces the restricted operators are maps valued in a linear space and the space of all the restricted operators is a linear space. Consequently, the notion of linear independence of the restricted operators is well defined---below this notion will be called {\em linear independence of the operators on} $\bld{\Psi}$.     

\begin{lm}
Let $\Cyl_K$ be a set of all cylindrical functions compatible with a set $K$ of independent d.o.f..  Assume that operators $\{\hat{\varphi}_1,\ldots,\hat{\varphi}_M\}\subset \hat{\F}$ act on elements of $\Cyl_K$ according to the formula in Assumption \ref{comp-f}. If $\{\hat{\varphi}_1,\ldots,\hat{\varphi}_M\}\subset \hat{\F}$ are linearly independent on a subset $\bld{\Psi}$ of $\Cyl_K$ then they are linearly independent on $K$.       
\label{cyl-K}
\end{lm}

\begin{pro}
Let $\Lambda$ be a subset of $\hat{\mathbf{F}}\times\mathbf{K}$ which satisfies Assumptions \ref{k-Lambda} and \ref{comp-f}. Then for every finite set $\{\hat{\varphi}_1,\ldots,\hat{\varphi}_M\}\subset\hat{\cal F}$ of linearly independent operators there exists a set $(\hat{F},K)\in\Lambda$ such that the operators are linearly independent on $K$.
\label{Lambda-pr}   
\end{pro} 

\noindent Both the lemma and the proposition are proven in \cite{q-stat}.

\begin{proof}[Proof of Lemma \ref{Lambda-dir}]

 The transitivity of the relation $\geq$ is obvious. Thus we have to prove  only that for any two elements $\lambda',\lambda\in\Lambda$ there exists $\lambda''\in\Lambda$ such that $\lambda''\geq\lambda'$ and $\lambda''\geq\lambda$. To achieve this we will refer to Lemma \ref{cyl-K} and Proposition \ref{Lambda-pr}. Therefore first we have to show that we are allowed to use them.

By virtue of Lemmas \ref{K-Kdg} and \ref{every-g} the set $\Lambda$ satisfies Assumption \ref{k-Lambda}. On the other hand, Equation \eqref{hphi-Psi} guarantees that every $\hat{\varphi}\in \hat{\F}$ acts on cylindrical functions compatible with $K_{\dot{\gamma}}$ according to the formula in Assumption \ref{comp-f} hence $\Lambda$ meets the assumption.    

Let us fix $\lambda'=(\hat{F}',K_{\dot{\gamma}'})$ and $\lambda=(\hat{F},K_{\dot{\gamma}})$. We define $\hat{F}_0$ as a linear subspace of $\hat{\F}$ spanned by elements of $\hat{F}'\cup\hat{F}$ and choose a basis $(\hat{\varphi}_1,\ldots,\hat{\varphi}_M)$ of $\hat{F}_0$. Proposition \ref{Lambda-pr} and Definition \ref{df-Lambda} of $\Lambda$ guarantee that there exists a speckled graph $\dot{\gamma}_0$ such that the operators $(\hat{\varphi}_1,\ldots,\hat{\varphi}_M)$ remain  linearly independent when restricted to $K_{\dot{\gamma}_0}$. Let $\dot{\gamma}''=(u'',\gamma'')$ be a speckled graph such that $(i)$ the number $3N$ of elements of $K_{\dot{\gamma}''}$ is greater than $\dim \hat{F}_0=M$ and $(ii)$ $\dot{\gamma}''\geq \dot{\gamma}_0,\dot{\gamma}',\dot{\gamma}$. By virtue of Lemma \ref{g'g-lin} d.o.f. in $K_{\dot{\gamma}_0}$ are cylindrical functions compatible with $K_{\dot{\gamma}''}$ and, according to Lemma \ref{cyl-K}, the operators $(\hat{\varphi}_1,\ldots,\hat{\varphi}_M)$ are linearly independent on $K_{\dot{\gamma}''}$.     

Consider now a matrix $G^0$ of components
\[
G^0_{\beta\alpha}:=\hat{\varphi}_{\beta}\kappa_{\alpha},
\]          
where $\{\kappa_1,\ldots,\kappa_{3N}\}=K_{\dot{\gamma}''}$. Clearly, the matrix has $M$ rows and $3N$ columns and because $(\hat{\varphi}_1,\ldots,\hat{\varphi}_M)$ are linearly independent on $K_{\dot{\gamma}''}$ its rank is equal $M<3N$. Using the following operations $(i)$ multiplying a row of by a non-zero number, $(ii)$ adding to a row a linear combination of other rows $(iii)$ reordering the rows and $(iv)$ reordering the columns we can transform the matrix $G^0$ to a matrix $G^1$ of the following form
\[
G^1=
\begin{pmatrix}
\mathbf{1} & G' 
\end{pmatrix},
\]             
where $\mathbf{1}$ is $M\times M$ unit matrix and $G'$ is a $M\times(3N-M)$ matrix. Note that the first three operations used to transform $G^0$ to $G^1$  correspond to a transformation of the basis $(\hat{\varphi}_1,\ldots,\hat{\varphi}_M)$ to an other basis $(\hat{\varphi}'_1,\ldots,\hat{\varphi}'_M)$ of $\hat{F}_0$, while the fourth operation corresponds to renumbering the d.o.f. in $K_{\dot{\gamma}''}$: $\kappa_{\alpha}\mapsto \kappa'_{\alpha}:=\kappa_{\sigma(\alpha)}$, where $\sigma$ is a permutation of the sequence $(1,\ldots,3N)$. Thus  
\[
G^1_{\beta\alpha}=\hat{\varphi}'_{\beta}\kappa'_{\alpha}.
\]

Let $\{\hat{\varphi}^0_1,\ldots,\hat{\varphi}^0_{3N}\}$ be operators constructed  with respect to $K_{\dot{\gamma}''}$ exactly as in the proof of Lemma \ref{every-g}. Then
\[
\hat{\varphi}^0_{\beta}\kappa'_{\alpha}=\delta_{\beta\alpha}.
\]   
Thus if
\[
\Big(\hat{\varphi}''_1,\ldots,\hat{\varphi}''_{3N}\Big):=\Big(\hat{\varphi}'_1,\ldots,\hat{\varphi}'_{M},\hat{\varphi}^0_{M+1},\ldots,\hat{\varphi}^0_{3N}\Big)
\]
then a $3N\times3N$ matrix $G=(G_{\beta\alpha})$ of components
\[
G_{\beta\alpha}:=\hat{\varphi}''_{\beta}\kappa'_{\alpha}
\] 
is of the following form
\[
G=
\begin{pmatrix}
\mathbf{1} & G'\\
\mathbf{0} & \mathbf{1}' 
\end{pmatrix},
\]  
where $\mathbf{0}$ is a zero $(3N-M)\times M$ matrix, and $\mathbf{1}'$ is a unit $(3N-M)\times(3N-M)$ matrix. The matrix $G$ is obviously non-degenerate which means in particular that the operators $(\hat{\varphi}''_1,\ldots,\hat{\varphi}''_{3N})$ are linearly independent.

To finish the proof it is enough to define 
\[
\hat{F}'':={\rm span}_{\mathbb{R}} \{\hat{\varphi}''_1,\ldots,\hat{\varphi}''_{3N}\}
\]
and $\lambda'':=(\hat{F}'',K_{\dot{\gamma}''})$.       

\end{proof}

\subsubsection{Checking Assumptions}

Now we have to check whether the directed set $(\Lambda,\geq)$ just constructed satisfies all Assumptions listed in Section \ref{ad-as}.

Proving Lemma \ref{Lambda-dir} we showed that $\Lambda$ satisfies Assumption \ref{k-Lambda}. Regarding Assumption \ref{f-Lambda} consider a set  $F_0=\{\varphi_1,\ldots,\varphi_N\}$ of momentum elementary d.o.f.. Let us fix $\varphi_i\in F_0$. Suppose that it is of the sort \eqref{phi-V} i.e. $\varphi_i=\varphi^{V_i}_{I_i}$ for some region $V_i$ and some $I_i\in\{1,2,3\}$. Then using a construction similar to that applied in the proof of Lemma \ref{every-g} one can find configurational d.o.f. $\{\kappa^{I}_{y_i}\}$ such that
\[
\hat{\varphi}^{V_i}_I\kappa^J_{y_i}=-\delta^J{}_I
\]    
for every $I,J=1,2,3$. Let $e_i$ be an edge such that $y_i\in e_i$. Then $\dot{\gamma}_i:=(\{y_i\},\{e_i\})$ is a speckled graph. As in the proof of Lemma \ref{every-g} one can find a face $S_i$ such that
\[
\hat{\varphi}^{S_i}_J\kappa^L_{e_i}=\delta^L{}_J.
\]     
for every $J,L=1,2,3$. Let 
\[
\hat{F}_i:={\rm span}\{\ \hat{\varphi}^{V_i}_I, \hat{\varphi}^{S_i}_J \ |\ I,J=1,2,3 \ \}. 
\]
Then $\hat{\varphi}_i\in\hat{F}_i$ and $(\hat{F}_i,K_{\dot{\gamma}_i})\in\Lambda$. Suppose now that $\varphi_i$ is of the sort \eqref{phi-S}. Then in a similar way one can construct an element  $(\hat{F}_i,K_{\dot{\gamma}_i})$ of $\Lambda$ such that  $\hat{\varphi}_i\in\hat{F}_i$. Since $\Lambda$ is a directed set there exists $(\hat{F},K_{\dot{\gamma}})\in \Lambda$ such that $(\hat{F},K_{\dot{\gamma}})\geq(\hat{F}_i,K_{\dot{\gamma}_i})$ for every $i=1,\ldots,N$. Taking into account Definition \ref{df-Lambda->} of the relation $\geq$ on $\Lambda$ we see that $\hat{F}$ contains all the operators $\{\hat{\varphi}_1,\ldots,\hat{\varphi}_N\}$. Thus Assumption \ref{f-Lambda} is satisfied.

Assumption \ref{RN} is satisfied by virtue of Lemma \ref{lm-Kug-xi}. We already concluded (proving Lemma \ref{Lambda-dir}) that $\Lambda$ meets Assumption \ref{comp-f}. Equations \eqref{ze-xi-const} and \eqref{hphiS-ke} guarantee that Assumption \ref{const} is satisfied and Definition \ref{df-Lambda} of $\Lambda$ ensures  that Assumption \ref{non-deg} holds. 

Consider now Assumption \ref{Q'=Q}. Let $(\hat{F},K_{\dot{\gamma}'}),(\hat{F},K_{\dot{\gamma}})$ be elements of $\Lambda$.  Recall that by virtue of Lemma \ref{K-Kdg} there exists $K_{\dot{\gamma}''}$ such that each d.o.f. in $K_{\dot{\gamma}}\cup K_{\dot{\gamma}'}$ is a linear combination of d.o.f. in $K_{\dot{\gamma}''}$. Suppose that $\Theta_{K_{\dot{\gamma}'}}=\Theta_{K_{\dot{\gamma}}}$. Then Lemma \ref{lm-Kug-xi} applied to $K_{\dot{\gamma}''}$ allows us set $\bar{K}=K_{\dot{\gamma}''}$ and $K=K_{\dot{\gamma}}$, $K'=K_{\dot{\gamma}'}$ in the following proposition \cite{q-stat}:  
\begin{pro}
Let $K,K'$ be sets of independent d.o.f. of $N$ and $N'$ elements respectively such that $\Theta_K=\Theta_{K'}$. Suppose that there exists a set $\bar{K}$ of independent d.o.f. of $\bar{N}$ elements such that the image of $\tilde{\bar{K}}$ is $\R^{\bar{N}}$ and  each d.o.f. in $K\cup K'$ is a linear combination of d.o.f. in $\bar{K}$. Then each d.o.f. in $K$ is a linear combination of d.o.f. in $K'$.  
\end{pro}
\noindent Thus each d.o.f. in $K_{\dot{\gamma}}$ is a linear combination of d.o.f. in $K_{\dot{\gamma}'}$. Then, as stated by Lemma \ref{g'g-lin}, $\dot{\gamma}'\geq \dot{\gamma}$ and, taking into account Definition \ref{df-Lambda->}, Assumption \ref{Q'=Q} follows.  

Assumption \ref{lin-comb} holds by virtue of Definition \ref{df-Lambda->} of the relation $\geq$ on $\Lambda$ and Lemma \ref{g'g-lin}, while Assumption \ref{FF'} is satisfied due to the definition. 

Thus the set $(\Lambda,\geq)$ satisfies all Assumptions. Consequently, it generates the space $\D$ of quantum states.

\subsection{The space $\D$ of quantum states for TEGR \label{D}}

Consider $\lambda=(\hat{F},K_{\dot{\gamma}})\in\Lambda$. The natural coordinates \eqref{lin-coor-0} define on the reduced configuration space $\Theta_{K_{\dot{\gamma}}}$ a measure
\begin{equation}
d\mu_\lambda:=dx_1\ldots dx_N.
\label{dmu-la}
\end{equation}
The measure provides a Hilbert space
\begin{equation}
\h_\lambda:=L^2(\Theta_{K_{\dot{\gamma}}},d\mu_\lambda)
\label{H-la}
\end{equation}
together with a set $\D_\lambda$ of all density operators (i.e. positive operators of trace equal $1$) on $\h_\lambda$. It was shown in \cite{q-stat} that given two elements $\la',\la$ of $\Lambda$ such that $\la'\geq\la$ there exists a distinguished projection $\pi_{\la\la'}$ from $\D_{\la'}$ onto $\D_\la$. The projection is defined as follows. 

If $\lambda'=(\hat{F}',K')\geq\lambda=(\hat{F},K)$ then every $\kappa_\alpha\in K$ is a linear combination of d.o.f. $\{\kappa'_1,\ldots,\kappa'_{N'}\}=K'$ (see Lemma \ref{g'g-lin}):
\begin{equation}
\kappa_\alpha=B^\beta{}_\alpha\kappa'_\beta,
\label{k-Bk'}
\end{equation}
where $(B^\beta{}_\alpha)$ are real numbers. This relation defines a linear projection $\pr_{KK'}:\Theta_{K'}\mapsto\Theta_K$:
\begin{equation}
\pr_{KK'}:=\tilde{K}^{-1}\circ (B\tilde{K}'),
\label{pr-KK}
\end{equation}
where $B\tilde{K}'$ means the action of the matrix $B=(B^\beta{}_\alpha)$ on the function $\tilde{K}'$ valued in the corresponding $\R^{N'}$. On the other hand, by virtue of Assumption \ref{comp-f} and \ref{const} every $\hat{\varphi}\in\hat{\cal F}$ defines a constant vector field 
\begin{equation}
\sum_\beta(\hat{\varphi}\kappa'_\beta)\partial_{x'_\beta}
\label{v-const}
\end{equation}
on $\Theta_{K'}$, where $(x'_\beta)$ are the natural coordinates on $\Theta_{K'}$. Since there is a natural one-to-one linear correspondence between constant vector fields on $\Theta_{K'}$ and points of this space every $\hat{\varphi}\in\hat{\cal F}$ distinguishes a point in $\Theta_{K'}$ which will be denoted by $[\hat{\varphi}]'$. The map $\hat{\varphi}\mapsto[\hat{\varphi}]'$ is linear and due to non-degeneracy of $(\hat{F}',K')$ its restriction to $\hat{F}'$ is invertible. Since $\hat{F}\subset \hat{F}'$ the image $[\hat{F}]'$ is a linear subspace of $\Theta_{K'}$ such that $\dim [\hat{F}]'=\dim\Theta_K$. It turns out that $\ker\pr_{KK'}\cap[\hat{F}]'=\varnothing$ hence
\[
\Theta_{K'}=\ker\pr_{KK'}\oplus[\hat{F}]'
\]                 
and 
\[
\omega_{\lambda'\lambda}:=\big(\pr_{KK'}\big|_{[\hat{F}]'}\big)^{-1}
\]
is a well defined linear isomorphism from $\Theta_K$ onto $[\hat{F}]'$. Using $\omega_{\lambda'\lambda}$ one pushes forward the measure $d\mu_\lambda$ obtaining a measure $d\mu_{\lambda'\lambda}$ on $[\hat{F}]'$ which allows to define a Hilbert space $\h_{\lambda'\lambda}$ over $[\hat{F}]'$---this Hilbert space is naturally isomorphic to $\h_\lambda$. There exists a unique measure $d\tilde{\mu}_{\lambda'\lambda}$ on $\ker\pr_{KK'}$ such that $d\mu_{\lambda'}=d\tilde{\mu}_{\lambda'\lambda}\times d\mu_{\lambda'\lambda}$---this measure provides a Hilbert space $\tilde{\h}_{\lambda'\lambda}$ such that
\[
\h_{\lambda'}=\tilde{\h}_{\lambda'\lambda}\ot\h_{\lambda'\lambda}.
\]           
Acting on $\rho'\in\D_{\lambda'}$ by the partial trace with respect to the Hilbert space $\tilde{\h}_{\lambda'\lambda}$ one gets a density operator on $\h_{\lambda'\lambda}$ which can be naturally mapped to an element $\rho\in\D_\lambda$---by definition 
\[
\pi_{\lambda\lambda'}\rho':=\rho.
\]
An important observation is that the projection $\pi_{\la\la'}$ is fully determined by the projection $\pr_{KK'}$ and the subspace $[\hat{F}]'$.   

It turns out that for every triplet $\la'',\la',\la\in\Lambda$ such that $\la''\geq\la'\geq\la$ the corresponding projections satisfy the following consistency condition
\begin{equation}
\pi_{\la\la''}=\pi_{\la\la'}\circ\pi_{\la'\la''},
\label{pipipi}
\end{equation}
which means that $\{\D_\la,\pi_{\la\la'}\}_{\la\in\Lambda}$ is a {\em projective family}. The space $\D$ of quantum states for a theory of the phase space $P\times\Theta$ is the {\em projective limit} of the family:
\[
\D:=\underleftarrow{\lim} \,\D_\lambda.
\] 

\subsection{$C^*$-algebra of quantum observables}

Let us recall briefly a construction of a $C^*$-algebra of quantum observables \cite{kpt} associated with the space $\D$. Denote by ${\cal B}_\lambda$ the $C^*$-algebra of bounded linear operators on the Hilbert space $\h_\lambda$ given by \eqref{H-la}. Each density operator  $\rho\in \D_\la$ defines an algebraic state (that is, linear $\C$-valued positive normed functional) on the algebra ${\cal B}_\la$ via a trace: 
\[
{\cal B}_\la\ni a \mapsto \tr(a\rho)\in\C.
\]          
This fact guarantees that for every pair $\lambda'\geq\lambda$ of elements of $\Lambda$ there exists a unique injective $*$-homomorphism $\pi^*_{\lambda'\lambda}:{\cal B}_\la\to{\cal B}_{\la'}$ dual to the projection $\pi_{\lambda\lambda'}:\D_{\lambda'}\to\D_\lambda$ in the following sense: for every $a\in{\cal B}_\la$ and every $\rho'\in \D_{\la'}$
\[
\tr(\pi^*_{\la'\la}(a)\rho')=\tr(a\,\pi_{\la\la'}(\rho')),
\]   
By virtue of  \eqref{pipipi} for every triplet $\la''\geq\la'\geq\la$ 
\[
\pi^*_{\la''\la}=\pi^*_{\la''\la'}\circ\pi^*_{\la'\la},
\]  
which means that $\{{\cal B},\pi^*_{\la'\la}\}$ is an inductive family of $C^*$-algebras associated with the projective family $\{\D_\lambda,\pi_{\lambda\lambda'}\}$.  Its inductive limit
\[
{\cal B}:= \underrightarrow{\lim}\,{\cal B}_\la
\]   
is naturally a unital $C^*$-algebra which can be interpreted as an algebra of quantum observables. It can be shown that each element $\rho$ of the space $\D$ defines an algebraic state on $\cal B$.    
     
\section{Action of spatial diffeomorphisms on $\D$ \label{diff-D}}

Since we would like to quantize TEGR in a background independent manner it is natural to follow LQG methods (see e.g. \cite{cq-diff,rev,rev-1}) to define an action of diffeomorphisms of the manifold $\Sigma$ on the space $\D$. Since $\Sigma$ represents a spatial slice of the original space-time the diffeomorphisms of $\Sigma$ can be regarded as spatial diffeomorphisms.    

\subsection{Action of diffeomorphisms on elementary d.o.f}

Let $\diff$ be a group of all analytic diffeomorphisms of $\Sigma$ which preserve the orientation of the manifold.  Consider an element $\tau$ of $\diff$. Since the fields $(\zeta_I,r_J,\xi^K,\theta^L)$ are differential forms on $\Sigma$ the diffeomorphism acts on them naturally as the pull-back $\tau^*$. Thus the pull-back define the action of the diffeomorphism on the phase space $P\times\Theta$. Since elementary d.o.f. in $\cal K$ and $\cal F$ are functions on the phase space it is natural to define an action of $\tau$ on  $\cal K$ and $\cal F$ as follows. Given $\kappa\in{\cal K}$, the result $\tau \kappa$ of the action of $\tau$ on $\kappa$ is a function on $\Theta$ such that
\begin{equation}
(\tau \kappa)(\theta)=\kappa(\tau^*\theta).
\label{tau-kappa}
\end{equation}
Similarly, given $\varphi\in{\cal F}$, the result $\tau \varphi$ of the action of $\tau$ on $\varphi$ is a function on $P$ such that
\[
(\tau \varphi)(p)=\kappa(\tau^*p).
\]
Obviously,
\begin{equation}
\begin{aligned}
\tau \kappa^I_y&=\kappa^I_{\tau(y)},&\tau \kappa^J_e&=\kappa^I_{\tau(e)},\\
\tau \varphi^V_I&=\varphi^{\tau(V)}_I,&\tau \varphi^S_J&=\varphi^{\tau(S)}_J,
\end{aligned}
\label{tau-kf}
\end{equation}
which mean that both sets $\cal K$ and $\cal F$ are preserved by the action of $\tau $. 


\subsection{Action of diffeomorphisms on reduced configuration spaces}

In the next step let us define an action of diffeomorphisms on reduced configuration spaces. Let us fix a finite set $K=\{\kappa_1,\ldots,\kappa_N\}$ of independent d.o.f. and a diffeomorphism $\tau$. Denote by $\tau  K$ the set $\{\tau \kappa_1,\ldots,\tau \kappa_N\}$. Moreover, let $\sim$ and $\sim_\tau$ be the equivalence relations on $\Theta$ defined by, respectively, $K$ and $\tau K$ (see Section \ref{fin-sets}) and let $[\theta]$ and $[\theta]_{\tau}$ denote equivalence classes of $\theta\in\Theta$ given by the corresponding relations. By definition $\theta\sim\theta'$ if and only if   $\kappa_\alpha(\theta)=\kappa_\alpha(\theta')$ for every $\kappa_\alpha\in K$. By virtue of \eqref{tau-kappa} the latter condition is satisfied if and only if for every $\tau \kappa_\alpha\in\tau K$
\[
\tau \kappa_\alpha(\tau^{-1*}\theta)=\tau \kappa_\alpha(\tau^{-1*}\theta').
\] 
This means that $\theta\sim\theta'$ if and only if $(\tau^{-1*}\theta)\sim_\tau(\tau^{-1*}\theta')$. Consequently, the following map
\[
\Theta_K\ni [\theta]\mapsto T_\tau([\theta]):=[\tau^{-1*}\theta]_\tau\in\Theta_{\tau K}
\]   
is well defined and is a bijection. 

Consider now the projections $\pr_K$ and $\pr_{\tau K}$ defined by \eqref{pr-K} and  the maps $\tilde{K}$ and $\widetilde{\tau K}$ defined by \eqref{k-inj}. We have
\[
\pr_{\tau K}(\tau^{-1*}\theta)=[\tau^{-1*}\theta]_\tau=T_\tau([\theta])=T_\tau(\pr_K(\theta)),
\]
hence
\begin{equation}
T^{-1}_\tau\circ\pr_{\tau K}=\pr_K\circ\tau^{*}.
\label{prt-Tpr}
\end{equation}
On the other hand,
\begin{multline*}
\tilde{K}([\theta])=(\kappa_1(\theta),\ldots,\kappa_N(\theta))=(\tau \kappa_1(\tau^{-1*}\theta),\ldots,\tau \kappa_N(\tau^{-1*}\theta))=\widetilde{\tau K}([\tau^{-1*}\theta]_\tau)=\\=\widetilde{\tau K}(T_\tau([\theta])),
\end{multline*}
hence
\begin{equation}
\tilde{K}=\widetilde{\tau K}\circ T_\tau.
\label{K-tKT}
\end{equation}
It follows from the this result that $\widetilde{\tau K}$ is a bijection onto $\R^N$ which means that $\tau K$ is a set of independent d.o.f. and $\Theta_{\tau K}$ is a reduced configuration space.

Let $(x_\alpha)$ be the natural coordinate frame \eqref{lin-coor-0} on $\Theta_K$  and let $(\bar{x}_\alpha)$ be the natural coordinate frame on $\Theta_{\tau K}$. Denote by $\{\partial_{x_\alpha}\}$ and $\{\partial_{\bar{x}_\alpha}\}$ vector fields defined by the coordinate frames on, respectively, $\Theta_K$ and $\Theta_{\tau K}$.  Equations \eqref{x-tilK} and \eqref{K-tKT} imply that the map $T_\tau$ when expressed in the coordinate frames is an identity map. Hence $T_{\tau *}\partial_{x_\alpha}=\partial_{\bar{x}_\alpha}$ and consequently
\begin{equation}
T^{-1*}_\tau(\partial_{x_\alpha}\psi)=\partial_{\bar{x}_\alpha}(T^{-1 *}_\tau\psi),
\label{Tx-xT}
\end{equation}
where $\psi$ is a function on $\Theta_K$.     

\subsection{Action of diffeomorphisms on cylindrical functions and momentum operators }

Now we will extend the action $\tau $ from $\cal K$ onto $\Cyl$: given $\Psi\in\Cyl$ we define
\[
(\tau \Psi)(\theta):=\Psi(\tau^*\theta)  
\]   
---this definition guarantees that $\tau $ acts linearly on $\Cyl$. Assume that $\Psi\in\Cyl$ is compatible with a set $K$ of independent d.o.f., that is, $\Psi=\pr_{K}\psi$. Then by virtue of \eqref{prt-Tpr}
\[
(\tau \Psi)(\theta)=\psi(\pr_K(\tau^*\theta))=\psi(T^{-1}_\tau(\pr_{\tau K}(\theta)))=[\pr^*_{\tau K}(T^{-1*}_\tau\psi)](\theta).
\]       
\begin{cor}
If $\Psi=\pr_K^*\psi$ then $\tau \Psi=\pr^*_{\tau K}(T^{-1*}_\tau\psi)$. If $\Psi$ is compatible with $K$ then $\tau \Psi$ is compatible with $\tau K$. 
\label{cor-tau-cyl}
\end{cor}
\noindent The corollary means that the action of $\tau$ preserves the space $\Cyl$---recall that every element of $\Cyl$ is a finite linear combination of functions such that each function is compatible with a set of independent d.o.f.. Thus $\Psi\mapsto\tau \Psi$ is a linear automorphism on $\Cyl$.

In the next step we define an action of diffeomorphisms on the linear space $\hat{\cal F}$ of the momentum operators: given $\hat{\varphi}\in\hat{\cal F}$, 
\[
(\tau \hat{\varphi})\Psi:=\tau \big(\hat{\varphi}(\tau^{-1} \Psi)).
\] 
It follows immediately from the definition that $\hat{\varphi}\mapsto\tau \hat{\varphi}$ is a linear map.

Let us now calculate $\hat{\varphi}^V_I(\tau\Psi)$. We know already that for every $\Psi\in\Cyl$ there exists a finite set $K\equiv K_{u,\gamma}$ and a complex function $\psi$ on $\Theta_K$ such that $\Psi=\pr_{K}\psi$. Obviously, 
\[
\tau K\equiv\tau K_{u,\gamma}=K_{\tau(u),\tau(\gamma)}
\]
Let $u=\{y_1,\ldots,y_N\}$ and let $(\bar{z}^{I}_i,\bar{x}^{J}_j)$ be the natural coordinates \eqref{lin-coor} on $\Theta_{\tau K}$. Then by virtue of \eqref{hat-zeta} and Corollary \ref{cor-tau-cyl}
\[
\hat{\varphi}^V_I(\tau \Psi)=\sum_{i=1}^N \varepsilon(V,\tau(y_{i}))\,\pr^*_{\tau K}(\partial_{\bar{z}^{I}_{i}}(T^{-1*}_\tau\psi))
\]
Using in turn \eqref{Tx-xT} and \eqref{prt-Tpr} we obtain
\[
\hat{\varphi}^V_I(\tau \Psi)=\tau \Big(\sum_{i=1}^N \varepsilon(V,\tau(y_{i}))\,\pr^*_{K}(\partial_{z^{I}_{i}}\psi))\Big),
\]
where $(z^I_i)$ are coordinates being a part of the natural coordinate frame on $\Theta_K$. Note now that by virtue of \eqref{Vy} 
\[
\varepsilon(V,\tau(y_i))=\varepsilon(\tau^{-1}(V),y_i)
\]
hence
\[
\hat{\varphi}^V_I(\tau \Psi)=\tau (\hat{\varphi}^{\tau^{-1}(V)}\Psi).
\]

We conclude that\footnote{Taking into account \eqref{tau-kf} we could define the action $\tau $ on $\hat{\cal F}$ requiring that $\tau \hat{\varphi}^V_I:=\hat{\varphi}^{\tau(V)}_I$ and $\tau \hat{\varphi}^S_J:=\hat{\varphi}^{\tau(S)}_J$, but then we could run into troubles with proving linearity of the action.}
\begin{align*}
\tau \hat{\varphi}^V_I&=\hat{\varphi}^{\tau(V)}_I,&\tau \hat{\varphi}^S_J&=\hat{\varphi}^{\tau(S)}_J
\end{align*}
---the latter equation can be proven similarly. The result means that the action of $\tau $ preserves the space $\hat{\cal F}$. Thus $\hat{\varphi}\mapsto\tau\hat{\varphi}$ is a linear automorphism on $\hat{\cal F}$.

\subsection{Action of diffeomorphisms on the directed set $(\Lambda,\geq)$ }

Given $\lambda=(\hat{F},K_{\dot{\gamma}})\in\Lambda$, we define
\[
\tau \lambda:=(\tau \hat{F},\tau K_{\dot{\gamma}}).
\] 
Let us prove now that this action preserves $\Lambda$ and the directing relation $\geq$ on it.

According to Definition \ref{df-Lambda} $\tau \lambda$ is an element of $\Lambda$ if and only if  $\tau K_{\dot{\gamma}}$ is a set of independent d.o.f. defined by a speckled graph and the pair $(\tau \hat{F},\tau K_{\dot{\gamma}})$ is non-degenerate. It is obvious that if $\dot{\gamma}=(u,\gamma)$ is a speckled graph then $\tau(\dot{\gamma})=(\tau(u),\tau(\gamma))$ is also a speckled graph. By virtue of \eqref{tau-kf} $\tau K_{\dot{\gamma}}=K_{\tau(\dot{\gamma})}$. On the other hand, since the action $\tau $ on $\hat{\cal F}$ is linear and invertible $\tau \hat{F}$ is a linear subspace of $\hat{\cal F}$ of the dimension equal to $\dim\hat{F}$. Therefore $\dim\tau \hat{F}$ is equal to the number of elements of $\tau K$. Let $(\hat{\varphi}_1,\ldots,\hat{\varphi}_N)$ be a basis of $\hat{F}$ and $K_{\dot{\gamma}}=\{\kappa_1,\ldots,\kappa_N\}$. Then $(\tau \hat{\varphi}_1,\ldots,\tau \hat{\varphi}_N)$ is a basis of $\hat{F}$ and $\tau K_{\dot{\gamma}}=\{\tau \kappa_1,\ldots,\tau \kappa_N\}$. We have
\[
\tilde{G}_{\beta\alpha}:= (\tau \hat{\varphi}_\beta)(\tau \kappa_\alpha)=\tau (\hat{\varphi}_\beta\kappa_\alpha)=\hat{\varphi}_\beta\kappa_\alpha=G_{\beta\alpha},
\] 
---here we used the fact that $\hat{\varphi}_\beta\kappa_\alpha$ is a constant cylindrical function. Thus non-degene\-ra\-cy of $(\tau \hat{F},\tau K_{\dot{\gamma}})$ follows from non-degeneracy of $(\hat{F},K)$. Consequently, $\tau \lambda\in\Lambda$.    

Consider now a pair $\lambda'=(\hat{F}',K_{\dot{\gamma}'})$ and $\lambda=(\hat{F},K_{\dot{\gamma}})$ such that $\lambda'\geq\lambda$. Using Definition \ref{df-Lambda->} and properties of the action of $\tau $ on $\hat{\cal F}$ we obtain  
\begin{align*}
&\tau \hat{F}'\supset \tau \hat{F},&&\tau(\dot{\gamma}')\geq\tau(\dot{\gamma}),
\end{align*}
which means that 
\[
\tau \lambda'=(\tau \hat{F}',\tau K_{\dot{\gamma}'})\geq \tau \lambda=(\tau \hat{F},\tau K_{\dot{\gamma}}).
\]
Consequently, the relation $\geq$ is preserved by the action $\tau $ on $\Lambda$.

We conclude that the directed set $(\Lambda,\geq)$ is preserved by the action of diffeomorphisms.    

\subsection{Action of diffeomorphisms on $\D$}

Consider $\lambda=(\hat{F},K_{\dot{\gamma}})\in\Lambda$. Recall that the map $T_\tau:\Theta_{K_{\dot{\gamma}}} \to\Theta_{\tau K_{\dot{\gamma}}}$ when expressed in the natural coordinate frames $(x_\alpha)$ on $\Theta_{K_{\dot{\gamma}}}$ and $(\bar{x}_\alpha)$ on $\Theta_{\tau K_{\dot{\gamma}}}$ is an identity map. This means that $T_\tau$ maps the measure $d\mu_\lambda$ on $\Theta_{K_{\dot{\gamma}}}$ defined by \eqref{dmu-la} to the measure $d\mu_{\tau \lambda}$ on $\Theta_{\tau K_{\dot{\gamma}}}$ defined analogously. Therefore the map
\begin{equation}
\h_\lambda\ni\psi\mapsto U_\tau\psi:=T^{-1*}_\tau\psi
\label{tau-Hl}
\end{equation}
is a {\em unitary} map onto $\h_{\tau \lambda}$. Consequently, 
\begin{equation}
\D_\lambda\ni\rho_\lambda\mapsto u_\tau\rho_\lambda:=U_\tau\rho_\lambda U^{-1}_\tau
\label{tau-Dl}
\end{equation}
is a map onto $\D_{\tau \lambda}$. 
     
Consider now an element $\rho$ of $\D$---by virtue of the definition of a projective limit $\rho$ is a family $\{\rho_\lambda\}_{\lambda\in\Lambda}$ such that  $\rho_\lambda\in\D_\lambda$ for every $\lambda$ and $\pi_{\lambda\lambda'}\rho_{\lambda'}=\rho_\lambda$ for every pair $\lambda'\geq\lambda$. It is natural to define an action of diffeomorphisms on $\rho$ as follows
\begin{equation}
\tau \rho:=\{u_\tau\rho_\lambda\}_{\lambda\in\Lambda},
\label{tau-rho}
\end{equation}
but is $\tau\rho$ an element of $\D$? Clearly, $\tau \rho\in\D$ if for every $\lambda'\geq\lambda$   
\[
\pi_{\bar{\lambda}\bar{\lambda}'}(u_\tau\rho_{\lambda'})=u_\tau\rho_\lambda,
\]
where we denoted $\bar{\lambda}\equiv\tau\lambda$ and $\bar{\lambda}'\equiv\tau\lambda'$  to keep the notation compact. Thus to prove that the action of diffeomorphisms on $\D$ preserves the space  we should show that
\begin{equation}
u^{-1}_\tau\circ\pi_{\bar{\lambda}\bar{\lambda}'}\circ u_\tau=\pi_{\lambda\lambda'}.
\label{upiu-pi}
\end{equation}

Assume that $\lambda'=(\hat{F}',K')\geq\lambda=(\hat{F},K)$. Recall the projection $\pi_{\lambda\lambda'}$ is determined by the projection $\pr_{KK'}$ and the subspace $[\hat{F}]'$ of $\Theta_{K'}$ (see Section \ref{D}). Similarly, the projection $\pi_{\bar{\lambda}\bar{\lambda'}}$ is constructed from $\pr_{\bar{K}\bar{K}'}$ and the subspace $\overline{[\tau\hat{F}]}{}'$ of $\Theta_{\tau K'}$, where $\bar{K}\equiv\tau K$, $\bar{K}'\equiv \tau K'$ and $\hat{\varphi}\mapsto \overline{[\hat{\varphi}]}{}'$ is the linear map from $\hat{\cal F}$ onto $\Theta_{\tau K'}$ defined in Section \ref{D}. Taking into account that the map $u_\tau$ appearing \eqref{upiu-pi} is defined by $T_\tau$ (see \eqref{tau-Dl} and \eqref{tau-Hl}) we conclude that to prove \eqref{upiu-pi} it is enough to show that
\begin{align}
T_\tau^{-1}\circ\pr_{\bar{K}\bar{K}'}\circ T_\tau&=\pr_{KK'},& T_\tau[\hat{F}]'=\overline{[\tau\hat{F}]}{}'.
\label{II}
\end{align}

Let us denote elements of $K$ and $K'$ as it was done in Section \ref{D}. It follows from \eqref{k-Bk'} that
\[
\tau\kappa_\alpha=B^\beta{}_\alpha\tau\kappa'_\beta,
\] 
hence by virtue of \eqref{pr-KK}
\[
\pr_{\bar{K}\bar{K}'}=\widetilde{\tau K}^{-1}\circ(B\widetilde{\tau K'}).
\]
Using \eqref{K-tKT} we obtain the first equation in \eqref{II}
\[
T_\tau^{-1}\circ\pr_{\bar{K}\bar{K}'}\circ T_\tau=T_\tau^{-1}\circ\widetilde{\tau K}^{-1}\circ (B\widetilde{\tau K'})\circ T_\tau=\tilde{K}^{-1}\circ(B\tilde{K}')=\pr_{KK'}.
\]

An operator $\hat{\varphi}\in\hat{F}$ defines on $\Theta_{K'}$ the constant vector field \eqref{v-const} which means that in the natural coordinate frame $(x'_\beta)$ the point $[\hat{\varphi}]'\in\Theta_{K'}$ is represented by $(\hat{\varphi}\kappa'_\beta)$. On the other hand, the operator $\tau\hat{\varphi}$ defines on $\Theta_{\tau K'}$ a constant vector field 
\[
\sum_\beta\big((\tau\hat{\varphi})(\tau\kappa'_\beta)\big)\partial_{\bar{x}'_\beta},     
\]
where $(\bar{x}'_\beta)$ are the natural coordinates on $\Theta_{\tau K'}$. Thus  the point $\overline{[\tau\hat{\varphi}]}{}'\in\Theta_{\tau K'}$ is represented by 
\[
\big((\tau\hat{\varphi})(\tau\kappa'_\beta)\big)=\big(\tau(\hat{\varphi}\kappa'_\beta)\big)=(\hat{\varphi}\kappa'_\beta)
\]
in the frame $(\bar{x}'_\beta)$. Since the map $T_\tau$ expressed in the coordinates $(x'_\beta)$ and $(\bar{x}'_\beta)$ is an identity we conclude that
\[
T_\tau([\hat{\varphi}]')=\overline{[\tau\hat{\varphi}]}{}'
\]   
which means that the second equation in \eqref{II} is true.

In this way we showed that $\tau\rho\in\D$, that is, that the action \eqref{tau-rho} of diffeomorphisms preserves the space $\D$.

\section{Other spaces of quantum states for a theory of the phase space $P\times \Theta$ \label{other}}

\subsection{Spaces built from other variables on the phase space}

A space $\bar{\D}$ of quantum states similar to $\D$ can be constructed by applying the natural description of the phase space \cite{q-suit}, that is, the description in terms of fields $(\theta^A,p_B)$ (see Section \ref{phsp}). In this case elementary d.o.f. are given by integrals of one-forms $(\theta^A)$ over edges and by integrals of two-forms $(p_B)$ over faces. To define a directed set $(\bar{\Lambda},\geq)$ which underlies the construction of $\bar{\D}$ it is enough to use the directed set of all usual (non-speckled) graphs. In other words, this construction is fully analogous to the construction of quantum states for DPG presented in \cite{q-stat}---the only difference between these two constructions is that in the case of $\bar{\D}$  the canonical variables are four one-forms $(\theta^A)$ and four two-forms $(p_B)$ while in \cite{q-stat} the canonical variables are one one-form and one two-form.    

Thus the construction of $\bar{\D}$ is simpler than that of $\D$. Unfortunately, the space $\bar{\D}$ possesses an undesirable property: as shown in \cite{q-suit} quantum states in $\bar{\D}$ correspond not only to elements of $\Theta$ by also to all quadruplets $(\theta^A)$ which define via \eqref{q} non-Riemannian metrics on $\Sigma$. Since we do not see any workable method which could distinguish in $\bar{\D}$ states corresponding only to elements of $\Theta$ we prefer to base the quantization of TEGR on the space $\D$.  

Let us emphasize that constructing the space $\D$ we never applied the fact that the variables $(\zeta_I,r_J,\xi^K,\theta^L)$ are defined by $\iota=\sgn$---this fact was used in the discussion in Section \ref{speckl}, but the only goal of this discussion was to show that the impossibility to approximate $\sgn(\theta^I)$ by means of  functions on $\Theta_{K_{\dot{\gamma}}}$ is not an obstacle for defining quantum geometry operators and quantum constraints as counterparts of classical constraints of TEGR and YMTM. In other words, the discussion concerned not the very construction of $\D$ but rather further applications of $\D$. This means that a space $\D_\iota$ of quantum states for TEGR can be built in the same way starting from any variables $(\zeta_{\iota I},r_J,\xi_\iota^K,\theta^L)$, however, as shown in \cite{ham-nv} the constraints of TEGR and YMTM derived in \cite{oko-tegr} and \cite{os} cannot be imposed on $\D_\iota$ unless $\iota=\sgn$ or $\iota=-\sgn$.  

In particular, the variables $(\zeta_{-s I},r_J,\xi_{-s}^K,\theta^L)$ given by $\iota=-\sgn$ can be used to construct a space $\D_{-s}$. By virtue of \eqref{new-old}
\begin{align}
\zeta_{sI}&=-\zeta_{-sI}, & \xi^K_s&=-\xi^K_{-s},
\label{s=--s}
\end{align}
where we used the original notation for the variables $(\zeta_{I},r_J,\xi^K,\theta^L)$ (see \eqref{simp-n}). These simple relations imply that the space $\D_{-s}$ and $\D$ are the same: $\D\equiv \D_{-s}$---a proof of this statement can be found in Appendix \ref{DD-s}.

\subsection{Hilbert spaces built from some almost periodic functions}

It was shown in \cite{q-stat} that for every theory for which it is possible to apply the general method presented in that paper to obtain a convex set of quantum states there exists another space of quantum states. This space is a Hilbert space built from almost periodic functions defined on those reduced configuration spaces which are isomorphic to $\R^N$. Thus in the case of TEGR there exist Hilbert spaces $\{\h_\iota\}$ and $\bar{\h}$: the former ones associated with the spaces $\{\D_\iota\}$ and the latter one with $\bar{\D}$. However, in order to proceed with the second step of the Dirac strategy we would have to define on such a Hilbert space operators corresponding to the constraints and we  expect this to be quite difficult. The source of the difficulty is the fact that on a Hilbert space of almost periodic functions on $\R^N$ the standard quantum operator of position is ill defined because an almost periodic function multiplied by a Cartesian coordinate on $\R^N$ is no longer an element of this Hilbert space. Since configurational elementary d.o.f. define Cartesian coordinates on configuration spaces we see that we would not be able to represent the configurational d.o.f. on $\h_\iota$ and $\bar{\h}$ by usual multiplication. To define an operator on $\h_\iota$ or $\bar{\h}$  corresponding to such a d.o.f. we would have to multiply the d.o.f. by a purely imaginary number and exponentiate the product. But taking into account the form of the constraints of TEGR \cite{oko-tegr,ham-nv} it is hard to expect that such ``exponentiated position operators'' can be used to represent the constraints. Thus the spaces $\{\h_\iota\}$ and $\bar{\h}$ do not seem to be very promising for canonical quantization of TEGR.                

\section{Discussion}

\subsection{General remarks}

The main results of this paper is the space $\D$ of quantum states and the related $C^*$-algebra $\cal B$ of quantum observables. The space $\D$ is not a Hilbert space but a convex set and each element of it naturally defines an algebraic state on the algebra, hence a Hilbert space can be obtained \cite{kpt} from any state in $\D$ and the algebra via the GNS construction. Although for every $\lambda\in\Lambda$ the space $\D_\lambda$ is a set of all density operators on the Hilbert space $\h_\lambda$ we do not expect that there exists a Hilbert space such that $\D$ is a set of density operators on it.    

The construction of $\D$ and $\cal B$ is based on the phase space $P\times \Theta$ described in Section \ref{phsp}. The elementary d.o.f. \eqref{k-y}, \eqref{k-e}, \eqref{phi-V} and \eqref{phi-S} used in the construction are defined as natural integrals of the canonical variables $(\zeta_I,r_J,\xi^K,\theta^L)$ being differential forms on the manifold $\Sigma$. Recall that the natural variables $(\theta^A,p_B)$ on the phase space are functions \eqref{old-new} of $(\zeta_I,r_J,\xi^K,\theta^L)$ involving the factor $\sgn(\theta^I)$ defined by \eqref{sgn-th}. Since the factor cannot be expressed or even approximated by the elementary d.o.f. (see Lemma \ref{theta-x3}) the spaces $\D$ and $\cal B$ may be useful only for a class of theories: the Hamiltonian (and possible constraints) of a theory belonging to this class when expressed in terms of the variables $(\zeta_I,r_J,\xi^K,\theta^L)$ may not depend on the factor. As shown in \cite{ham-nv}, both TEGR and YMTM belong to this class.         

\subsection{Diffeomorphism invariant states}

Since $\D$ is a space of kinematic quantum states to proceed further with the canonical quantization of TEGR we have to find a procedure by means of which we could single out physical quantum states for TEGR---an outline of such procedure was presented in \cite{q-stat}. Because TEGR is a diffeomorphism invariant theory it is reasonable to require that each physical state is invariant with respect to the natural action of the spatial diffeomorphisms on the space $\D$ defined in Section \ref{diff-D} as it is required in the case of LQG \cite{cq-diff,rev,rev-1}. Existence of such states in $\D$ and possible uniqueness are open questions---at this moment it is difficult to predict whether a theorem of existence and uniqueness of such a state analogous to those presented in \cite{lost,fl} can be proven; let us only note that in the case of a space of quantum states for DPG constructed in \cite{oko-ncomp} there are plenty of diffeomorphism invariant states, however that construction does not follow the general pattern described in \cite{q-stat} and differs significantly from both the present construction of $\D$ and the construction of the space of quantum states for DPG described in \cite{q-stat}.

\subsection{The space $\D$ versus the kinematic Hilbert space of LQG}

The space $\D$ is a space of kinematic quantum states meant to serve as an element of a background independent canonical quantization of general relativity (GR) in the teleparallel formulation. Let us compare the space with its counterpart in LQG since LQG is a result of a background independent canonical quantization of an other formulation of GR. 

The counterpart is the Hilbert space $\h_{\rm LQG}$ defined as a space of some wave functions. These wave functions are defined on a space $\Abar$ of so called generalized $SU(2)$-connections \cite{proj} over a three-dimensional manifold $\Sigma$ and the scalar product on $\h_{\rm LQG}$ is defined by an integral with respect to the Ashtekar-Lewandowski (AL) measure $d\mu_{\rm AL}$ \cite{al-hoop} on $\Abar$:
\[
\h_{\rm LQG}:=L^2(\Abar,d\mu_{\rm AL}). 
\]

Alternatively, the space $\h_{\rm LQG}$ can be seen as the inductive limit of an inductive family of Hilbert spaces $\{\h_\gamma,p_{\gamma'\gamma}\}$ labeled by the directed set of (usual) graphs in $\Sigma$ (for some details of this alternative description see \cite{oko-ncomp}). Each Hilbert space $\h_\gamma$ is defined as follows: given graph $\gamma$, one reduces the Hamiltonian configuration space $\A$ of LQG being the space of all $SU(2)$-connections over $\Sigma$ obtaining a reduced configuration space $\A_\gamma$ isomorphic to $SU(2)^N$, where $N$ is the number of edges of $\gamma$. Next, one defines
\[
\h_\gamma:=L^2(\A_\gamma,d\mu_{\gamma}),
\]              
where $d\mu_{\gamma}$ is a measure on $\A_\gamma$ given uniquely by the normed Haar measure on $SU(2)^N$.

It is easy to find some close similarities between elements of the construction of $\h_{\rm LQG}$ and those of the construction of $\D$: $\A$ corresponds to the Hamiltonian configuration space $\Theta$, the spaces $\{A_\gamma\}$ are counterparts of  the reduced configuration spaces $\{\Theta_{K_{\dot{\gamma}}}\}$, likewise the Hilbert spaces $\{\h_\gamma\}$ are counterparts of the spaces $\{\h_\lambda\}$. Note also that the measure $d\mu_\lambda$ given by \eqref{dmu-la} which defines $\h_\la$ via \eqref{H-la} is in fact a Haar measure on $\Theta_{K_{\dot{\gamma}}}$ (the latter space being a real linear space is naturally a Lie group). Moreover, as shown in \cite{q-stat} for the space $\Theta$ there exists a space $\bar{\Theta}$ related to $\Theta$ in the same way as $\Abar$ is related to $\A$.                   

One may ask now why we did not define a Hilbert space for TEGR in the same way as the space $\h_{\rm LQG}$ is defined? The answer is very simple: each space $\A_\gamma$ is {\em compact} and this fact enables to define the AL measure on $\Abar$ and, alternatively, it enables to define the embeddings $\{p_{\gamma'\gamma}:\h_\gamma\to\h_{\gamma'}\}$ which allow to ``glue'' the Hilbert spaces $\{\h_\gamma\}$ into $\h_{\rm LQG}$ via the inductive limit. On the other hand, every space $\Theta_{K_{\dot{\gamma}}}$ is {\em non-compact} and this fact turns out to be an obstacle for defining a measure on $\bar{\Theta}$ as a counterpart of the AL measure and, alternatively, it turns out to be an obstacle for defining embeddings $p_{\lambda'\lambda}:\h_\lambda\to\h_{\lambda'}$ which would allows us to ``glue'' the spaces  $\{\h_\lambda\}$  into a larger one by means of an inductive limit. In other words, non-compactness of the spaces $\{\Theta_{K_{\dot{\gamma}}}\}$ precludes the use of the inductive techniques but on the other hand linearity of the spaces allows us to apply the projective techniques according to the original idea by Kijowski \cite{kpt}.        

Note however that the compactness of the spaces $\{\A_\gamma\}$ is in fact obtained by means of a reduction of the natural Lorentz symmetry of GR done at the level of the classical theory---this symmetry is reduced to its ``sub-symmetry'' described by the group of three-dimensional rotations. Technically it is achieved by a passage from the complex Ashtekar-Sen connections \cite{a-var-1,a-var-2} of the non-compact structure group $SL(2,\C)$ to the real Ashtekar-Barbero connections \cite{barb} of the compact structure group $SU(2)$. Let us emphasize that the construction of $\D$ does not require any reduction of the Lorentz symmetry of the classical theory, however it is still to early to claim that there are no obstacles for defining local Lorentz transformations on $\D$---this issue needs to be analyzed carefully.

Let us finally mention an important difference between the spaces $\D$ and $\h_{\rm LGQ}$ (for a similar discussion see \cite{oko-ncomp}). Both spaces $\D$ and $\h_{\rm LGQ}$  are built from some spaces associated with (speckled or usual) graphs in $\Sigma$: in the former case these spaces are $\{\D_\lambda\}$ ($\lambda=(F,K_{\dot{\gamma}})$, but $\D_\lambda$ does not depend actually on the space $F$), in the latter one these spaces are $\{\h_\gamma\}$. Since $\D$ is the projective limit of $\{\D_\lambda\}$ each state $\rho\in\D$ is a collection $\{\rho_\lambda\}$ of states such that $\rho_\lambda\in\D_\lambda$. This means that, given $\lambda$, the state $\rho_\lambda$ contains only a partial information about $\rho$ and therefore it can be treated merely as {\em an approximation} of $\rho$ \cite{kpt}. On the other hand, in the case of $\h_{\rm LQG}$ defined as the inductive limit of $\{\h_\gamma\}$ for every graph $\gamma$ there exists a canonical embedding $p_\gamma:\h_\gamma\to\h_{\rm LQG}$ and consequently each element of $\h_\gamma$ can be treated as a rightful element of $\h_{\rm LQG}$.                       

\paragraph{Acknowledgments} This work was partially supported by the grant N N202 104838 of Polish Ministerstwo Nauki i Szkolnictwa Wy\.zszego.
                                                   
\appendix

\section{The spaces $\D$ and $\D_{-s}$ are the same \label{DD-s}}

The space $\D_{-s}$ is built exactly in the same way as the space $\D$ is, the only difference is that the starting point of the construction of $\D_{-s}$ are the variables $(\zeta_{-s I},r_J,\xi_{-s}^K,\theta^L)$. Let us then trace all steps of both constructions noting differences and similarities between them. In what follows the variable $(\zeta_{I},r_J,\xi^K,\theta^L)$ will be called {\em  first variables} while $(\zeta_{-s I},r_J,\xi_{-s}^K,\theta^L)$ will be called {\em second variables}.  

Elementary d.o.f. $\bar{\kappa}^I_y,\bar{\kappa}^J_e,\bar{\varphi}^V_K,\bar{\varphi}^S_L$ defined in an obvious way by the second variables  are related to the d.o.f. originating from the first ones as follows 
\begin{equation}
\begin{aligned}
\bar{\kappa}^I_y&=-\kappa^I_y,&\bar{\kappa}^J_e&=\kappa^J_e,
&\bar{\varphi}^V_K&=-\varphi^V_K,&\bar{\varphi}^S_L&=\varphi^S_L
\end{aligned}    
\label{bk-k}
\end{equation}
---see \eqref{s=--s}.

Let $\bar{K}_{u,\gamma}$ be a set of d.o.f. $\{\bar{\kappa}^I_y,\bar{\kappa}^J_e\}$ distinguished by the finite set $u\subset\Sigma$ and the graph $\gamma$. Using the formulae \eqref{bk-k} it is easy to realize that although $K_{u,\gamma}\neq\bar{K}_{u,\gamma}$ the equivalence relations $\sim_{K_{u,\gamma}}$ and $\sim_{\bar{K}_{u,\gamma}}$ coincide hence
\begin{align}
\Theta_{K_{u,\gamma}}&=\Theta_{\bar{K}_{u,\gamma}}, &\pr_{K_{u,\gamma}}&=\pr_{\bar{K}_{u,\gamma}}
\label{T-bT}
\end{align}
and both maps $\tilde{K}_{u,\gamma}$ and $\tilde{\bar{K}}_{u,\gamma}$ are bijections onto the same $\R^{3(N+M)}$, where $N$ is the number of points of $u$ and $M$ is the number of edges of $\gamma$. Note that this $\R^{3(N+M)}$ can be naturally decomposed into a direct sum $\R^{3N}\oplus\R^{3M}$---the first term in the sum is constituted by values of d.o.f. defined by points of $u$, while the second one by d.o.f. given by edges of $\gamma$. Consider now a $3(N+M)\times 3(N+M)$ matrix 
\begin{equation}
\mathbf{I}=
\begin{pmatrix}
-\mathbf{1}&\mathbf{0}\\
\mathbf{0}&\mathbf{1}
\end{pmatrix}
\label{matr-I}
\end{equation}
being a block matrix with respect to the decomposition $\R^{3N}\oplus\R^{3M}$, where $\mathbf{1}$ is a unit matrix. It follows from  \eqref{bk-k} that if the order of elements $\{\kappa_1,\ldots,\kappa_N\}$ of ${K}_{u,\gamma}$ corresponds naturally\footnote{The order of elements of ${K}_{u,\gamma}$  corresponds naturally to the order of elements of $\bar{K}_{u,\gamma}$ if for every $\alpha\in\{1,\ldots,N\}$ either $(i)$ $\kappa_\alpha=\kappa^I_y$ and $\bar{\kappa}_\alpha=\bar{\kappa}^I_y$ for $y\in u$ or $(ii)$ $\kappa_\alpha=\kappa^I_e$ and $\bar{\kappa}_\alpha=\bar{\kappa}^I_e$ for an edge $e$ of $\gamma$.} to the order of elements $\{\bar{\kappa}_1,\ldots,\bar{\kappa}_N\}$ of $\bar{K}_{u,\gamma}$ and if the first $3N$ elements of both sets are defined by points of $u$  then
\begin{equation}
\tilde{K}_{u,\gamma}=\mathbf{I}\tilde{\bar{K}}_{u,\gamma},
\label{K-IK}
\end{equation}
which means that both linear structures defined on $\Theta_{K_{u,\gamma}}$ by $\tilde{K}_{u,\gamma}$ and $\tilde{\bar{K}}_{u,\gamma}$ coincide. Let $(z^I_{i},x^J_{j})$ be the natural coordinates \eqref{lin-coor} on $\Theta_{K_{u,\gamma}}$ and $(\bar{z}^I_{i},\bar{x}^J_{j})$ the natural coordinates on $\Theta_{\bar{K}_{u,\gamma}}$. Obviously,
\begin{align}
z^I_{i}&=-\bar{z}^I_{i}, & x^J_{j}&=\bar{x}^J_{j}.
\label{z--bz}
\end{align}

Taking into account \eqref{T-bT} and \eqref{z--bz} we conclude that each cylindrical function $\Psi$ compatible with $K_{u,\gamma}$ is compatible with $\bar{K}_{u,\gamma}$ and vice versa. Consequently, the space $\Cyl$ defined by the first variables coincides with that defined by the second ones.            

Let $\hat{\bar{\varphi}}^V_I$ be a momentum operator defined by the d.o.f. $\bar{\varphi}^V_I$. Consider a cylindrical function $\Psi=\pr^*_{\bar{K}_{u,\gamma}}\psi$. By virtue of \eqref{hat-zeta-1}  
\begin{equation}
\hat{\bar{\varphi}}^V_I\Psi=\sum_{L=1}^3\sum_{l=1}^M \pr^*_{\bar{K}_{u,\gamma}}(\partial_{\bar{z}^L_{l}}\psi)\{\bar{\varphi}^V_I,\bar{\kappa}^L_{y_{l}}\}
\label{hat-bar-zeta}
\end{equation}
It follows from \eqref{z--bz} that $\partial_{\bar{z}^I_i}=-\partial_{z^I_i}$. Using this fact, \eqref{bk-k} and \eqref{T-bT} we obtain 
\[
\hat{\bar{\varphi}}^V_I\Psi=\sum_{L=1}^3\sum_{l=1}^M \pr^*_{{K}_{u,\gamma}}(-\partial_{{z}^L_{l}}\psi)\{{\varphi}^V_I,{\kappa}^L_{y_{l}}\}=-\hat{\varphi}^V_I\Psi
\]
since $\Psi=\pr^*_{K_{u,\gamma}}\psi$. Consequently,
\[
\hat{\bar{\varphi}}^V_I=-\hat{\varphi}^V_I.
\]
This result allows us to conclude that the linear space $\hat{\cal F}$ defined by the first variables coincides with that defined by the second ones. 

Consider now the directed sets $(\Lambda,\geq)$ and $(\bar{\Lambda},\geq)$ given by, respectively, the first and the second variables. Let $\hat{F}$ be a finite dimensional subspace of $\hat{\cal F}$ and let $(\hat{\varphi}_1,\ldots,\hat{\varphi}_N)$ be a basis of $\hat{F}$. Moreover, let $K_{\dot{\gamma}}=\{\kappa_1,\ldots,\kappa_N\}$ and $\bar{K}_{\dot{\gamma}}=\{\bar{\kappa}_1,\ldots,\bar{\kappa}_N\}$ be sets of independent d.o.f. Assume that the order of elements of $\bar{K}_{\dot{\gamma}}$ corresponds naturally to the order of elements of ${K}_{\dot{\gamma}}$ and the first $n<N$ elements of both sets correspond to points of $u$, where $(u,\gamma)=\dot{\gamma}$. Then
\[
\bar{G}_{\beta\alpha}:=\hat{\varphi}_\beta\bar{\kappa}_\alpha=
\begin{cases}
-\hat{\varphi}_\beta{\kappa}_\alpha&\text{if $\alpha\leq n$},\\
\hat{\varphi}_\beta {\kappa}_\alpha&\text{otherwise} 
\end{cases},
\]
hence
\[
\bar{G}=\mathbf{I} G,
\]
where $\bar{G}=(\bar{G}_{\beta\alpha})$, ${G}=({G}_{\beta\alpha})$ and $\mathbf{I}$ is the matrix \eqref{matr-I}. This means that the matrix $\bar{G}$ is non-degenerate if and only if  $G$ is non-degenerate. Consequently, the pair $(\hat{F},\bar{K}_{\dot{\gamma}})\in\bar{\Lambda}$ if and only if $(\hat{F},{K}_{\dot{\gamma}})\in{\Lambda}$ and the map
\[
{\Lambda}\ni{\lambda}\equiv(\hat{F},{K}_{\dot{\gamma}})\mapsto r(\lambda):=(\hat{F},\bar{K}_{\dot{\gamma}})\in\bar{\Lambda}
\]      
is a bijection. It is easy to check that the bijection preserves the directing relation.

Thus although ${\Lambda}\neq\bar{\Lambda}$ the directed sets are naturally isomorphic. 

Let $\bar{\lambda}=r(\lambda)$, where $\la=(\hat{F},K_{\dot{\gamma}})$. Because of \eqref{z--bz} the measures $d\mu_\lambda$ and $d\mu_{\bar{\la}}$ on $\Theta_{K_{\dot{\gamma}}}$ coincide hence
\begin{align}
\h_\lambda&=\h_{\bar{\lambda}}, & \D_{\la}&=\D_{\bar{\la}}.
\label{d-d-s}
\end{align}

Consider now $\lambda$ and $\bar{\lambda}$ as above and $\bar{\la}'=r(\la')$, where $\la'=(\hat{F}',K_{\dot{\gamma}'})$ and assume that $\la'\geq\la$; then $\bar{\la}'\geq\bar{\la}$. Our goal now is to show that 
\begin{equation}
\pi_{\la\la'}=\pi_{\bar{\la}\bar{\la'}}.
\label{pi-pi-s}
\end{equation}
To reach the goal it is enough to prove that
\begin{align}
\pr_{K_{\dot{\gamma}}K_{\dot{\gamma}'}}&=\pr_{\bar{K}_{\dot{\gamma}}\bar{K}_{\dot{\gamma}'}}, & [\hat{F}]'&=\overline{[\hat{F}]}{}',
\label{pr-pr-s}
\end{align}
where $\hat{\varphi}\mapsto\overline{[\hat{\varphi}]}{}'$ is the linear map from $\hat{\cal F}$ onto $\Theta_{\bar{K}_{\dot{\gamma}'}}$ defined by the second variables (see Section \ref{D}).   

To prove the first equation in \eqref{pr-pr-s} assume that the order of elements $\{\kappa_1,\ldots,\kappa_N\}$ of $K_{\dot{\gamma}}$ corresponds to the order of elements $\{\bar{\kappa}_1,\ldots,\bar{\kappa}_N\}$ of $\bar{K}_{\dot{\gamma}}$ and that the first $n<N$ elements of both sets are given by points of $u$, where $(u,\gamma)=\dot{\gamma}$. We impose an analogous requirement on the order of elements $\{\kappa'_1,\ldots,\kappa'_{N'}\}$ of $K_{\dot{\gamma}'}$ and $\{\bar{\kappa}'_1,\ldots,\bar{\kappa}'_{N'}\}$ of $\bar{K}_{\dot{\gamma}'}$ assuming that the first $n'$ elements of both sets are defined by points of $u'$. Since $\dot{\gamma}'\geq\dot{\gamma}$ by virtue of Lemma \ref{g'g-lin}  
\begin{align*}
\kappa_\alpha&=B^\beta{}_\alpha\kappa'_\beta,&\bar{\kappa}_\alpha&=\bar{B}^\beta{}_\alpha\bar{\kappa}'_\beta,
\end{align*}
where ${B}^\beta{}_\alpha,\bar{B}^\beta{}_\alpha$ are real numbers. Note, that d.o.f. of the sort \eqref{k-y} are linearly independent of d.o.f. of the sort \eqref{k-e} and vice versa. This means that
${B}^\beta{}_\alpha=0$ if $(i)$ $\beta>n'$ and $\alpha \leq n$ or $(ii)$ $\beta\leq n'$ and $\alpha > n$. Of course, $\bar{B}^\beta{}_\alpha=0$ exactly in the same cases. These facts together with \eqref{bk-k} mean that
\[
{B}^\beta{}_\alpha=\bar{B}^\beta{}_\alpha.
\]     

Let $\mathbf{I}'$ be an $N'\times N'$ matrix constructed analogously to the matrix \eqref{matr-I}. Then using \eqref{pr-KK} and \eqref{K-IK} we obtain
\[
\pr_{K_{\dot{\gamma}}K_{\dot{\gamma}'}}=\tilde{K}_{\dot{\gamma}}^{-1}\circ (B\tilde{K}_{\dot{\gamma}'})=\tilde{\bar{K}}_{\dot{\gamma}}^{-1}\mathbf{I}\circ (B\mathbf{I}'\tilde{\bar{K}}_{\dot{\gamma}'})=\tilde{\bar{K}}_{\dot{\gamma}}^{-1}\circ (\bar{B}\tilde{\bar{K}}_{\dot{\gamma}'})= \pr_{\bar{K}_{\dot{\gamma}}\bar{K}_{\dot{\gamma}'}}.
\]     

To prove the second equation in \eqref{pr-pr-s} we assume that elements of $K_{\dot{\gamma}'}$ and $\bar{K}_{\dot{\gamma}'}$ are ordered as above and  note that due to \eqref{bk-k} and \eqref{z--bz} each $\hat{\varphi}\in\hat{\cal F}$ defines the same constant vector field on $\Theta_{K_{\dot{\gamma}'}}=\Theta_{\bar{K}_{\dot{\gamma}'}}$ regardless we use the first or the second variables:
\[
\sum_{\beta=1}^{N'}(\hat{\varphi}\kappa'_\beta)\partial_{x'_\beta}=\sum_{\beta=1}^{n'}\big(\hat{\varphi}(-\kappa'_\beta)\big)(-\partial_{x'_\beta})+\sum_{\beta=n'+1}^{N'}(\hat{\varphi}\kappa'_\beta)\partial_{x'_\beta}=\sum_{\beta=1}^{N'}(\hat{\varphi}\bar{\kappa}'_\beta)\partial_{\bar{x}'_\beta}
\]   
---here $(x'_\alpha)$ are the natural coordinates \eqref{lin-coor-0} on $\Theta_{K_{\dot{\gamma}'}}$, and $(\bar{x}'_\alpha)$ are the natural coordinates on $\Theta_{\bar{K}_{\dot{\gamma}'}}$. Therefore $[\hat{\varphi}]'=\overline{[\hat{\varphi}]}{}'$ and the second equation in \eqref{pr-pr-s} follows. 

The fact that $\Lambda$ and $\bar{\Lambda}$ are isomorphic, the second equation \eqref{d-d-s} and Equation \eqref{pi-pi-s} just proven mean that $\D=\D_{-s}$.



\end{document}